\documentclass[USenglish,onecolumn]{article}
\usepackage[utf8]{inputenc}

\def\thick#1{\hbox{\rlap{$#1$}\kern0.25pt\rlap{$#1$}\kern0.25pt$#1$}}

\def\smbalpha{\boldsymbol{{\scriptstyle{\alpha}}}}

\def\smbalpha{\widehat{\smbalpha}}

\def\hbar{\bar{ h}}

\def\E{\mbox{E}}

\def\inv{^{-1}}

\def\mybox#1{\vskip1mm \begin{center}
        \hspace{.0\textwidth}\vbox{\hrule\hbox{\vrule\kern6pt
\parbox{.9\textwidth}{\kern6pt#1\vskip6pt}\kern6pt\vrule}\hrule}
        \end{center} \vskip-5mm}
\def\lboxit#1{\vbox{\hrule\hbox{\vrule\kern6pt
      \vbox{\kern6pt#1\vskip6pt}\kern6pt\vrule}\hrule}}

\def\thickboxit#1{\vbox{{\hrule height 1mm}\hbox{{\vrule width 1mm}\kern6pt
          \vbox{\kern6pt#1\kern6pt}\kern6pt{\vrule width 1mm}}
               {\hrule height 1mm}}}

\def\fat#1{\hbox{\rlap{$#1$}\kern0.25pt\rlap{$#1$}\kern0.25pt$#1$}}

\newcommand{\I}{{\sf I}}

\usepackage{amsthm}
\usepackage{xcolor}
\newtheorem{remark}{Remark}
\newtheorem{proposition}{Proposition}
\usepackage{graphicx,verbatim,array,multicol, subfigure, color, lscape, mathrsfs, dsfont}
\usepackage{psfrag, amsmath, amsfonts, epsfig, fancybox, setspace, soul, amsthm, longtable, threeparttable, threeparttablex, makecell, xcolor, pdflscape, tabu, booktabs, natbib, amssymb}
\usepackage{booktabs}  
\usepackage{multirow}  
\usepackage{array}     
\usepackage[margin=1in]{geometry}  
\usepackage{authblk}
\usepackage[normalem]{ulem}
\usepackage{cancel}

\def\P{\text{Pr}}

\def\inv{{-1}}

\def\E{\mathbb{E}}
\def\P{\mathbb{P}}
\def\V{\mathbb{V}}
\def\I{\mathbb{I}}

\newtheorem{assumption}{Assumption}
\newtheorem{lemma}{Lemma}
\newtheorem{corollary}{Corollary}
\newtheorem{definition}{Definition}

\newtheorem{intassumption}{Assumption}
\numberwithin{intassumption}{assumption}
\makeatletter
\@addtoreset{intassumption}{assumption}
\makeatother

\usepackage{setspace}
\begin{document}

\title{State policy heterogeneity analyses: considerations and proposals}

\author[1]{Max Rubinstein\thanks{Corresponding author: mrubinstein@rand.org}}
\author[1]{Megan S. Schuler}
\author[2]{Elizabeth A. Stuart}
\author[1]{Bradley D. Stein}
\author[1]{Max Griswold}
\author[3]{Elizabeth M. Stone}
\author[1]{Beth Ann Griffin}

\affil[1]{RAND Corporation}
\affil[2]{Department of Biostatistics, Bloomberg School of Public Health, Johns Hopkins University}
\affil[3]{Rutgers Institute for Health, Health Care Policy, and Aging Research, Rutgers University}

\maketitle
\begin{abstract}
State-level policy studies often conduct heterogeneity analyses that quantify how treatment effects vary across state characteristics. These analyses may be used to inform state-specific policy decisions, or to infer how the effect of a policy changes in combination with another policy or state characteristics. However, in state-level settings with varied contexts and policy landscapes, multiple versions of similar policies, and differential policy implementation, the causal quantities implicitly targeted by these analyses may not align with the inferential goals. This paper clarifies these issues by distinguishing several causal estimands relevant to heterogeneity analyses in state-policy settings, including state-specific treatment effects, conditional average treatment effects, and controlled direct effects. We argue that the conditional average treatment effect is often the easiest to identify and estimate, but may not be the most policy relevant target of inference. Moreover, the widespread practice of coarsening distinct policies or implementations into a single treatment indicator further complicates the interpretation of these analyses. Motivated by these limitations, we propose bounding state-specific treatment effects as an alternative inferential goal, yielding ranges for each state's policy effect under explicit assumptions that quantify deviations from the ideal identifying conditions. These bounds target a well-defined and policy-relevant quantity, what the effect would be for specific states. We develop this approach within a difference-in-differences framework and discuss how the required sensitivity parameters may be informed using pre-treatment data. Through simulations we demonstrate that bounding state-specific effects can more reliably determine the sign of the state-specific effects than conditional average treatment effect estimates. We then illustrate this method to examine the effect of the Affordable Care Act Medicaid expansion on high-volume buprenorphine prescribing.
\end{abstract}

\section{Introduction}

State-policy analyses often use differential adoption of similar policies across U.S. states to quantify the average effect of different policies. Beyond estimating average effects, another important aim is to understand effect \emph{heterogeneity}; that is, to understand whether and why effects differ across states. These analyses may be used to inform context-specific decision making, including whether specific states are likely to benefit from a policy, or whether some combination of policies is more effective than another (see, e.g., \cite{garthwaite2019all, antonelli2024autoregressive, smart2024investigating}). A common strategy to estimate effect heterogeneity is by regressing outcomes on policy indicators and interactions between the policy indicators and observed state characteristics using generalized linear models. Analysts may then use the coefficients on the interaction terms both to quantify effects among states with different characteristics, and to determine which factors may drive apparent effect heterogeneity.

We argue that in state-level settings these kinds of heterogeneity analyses often target causal quantities that do not align with the questions researchers and policymakers intend to answer. For example, policy-relevant questions such as ``What was or would be the effect of adopting this policy in \emph{this state}?'' or ``Would this policy work better if combined with another policy?'' are not generally answered by the types of associational statements, such as ``Effects were estimated to be larger/smaller in states with these characteristics'', typically made in heterogeneity analyses \cite{hettinger2025causal, vegetabile2021distinction}. That is, effect heterogeneity analyses typically only reveal \emph{associations} between causal effects and state characteristics, not causal effects themselves. These interpretative challenges are compounded in state policy settings by the common practice of collapsing distinct policies or implementations into a single treatment indicator, or treatment coarsening, further obscuring the link between heterogeneity analyses and policy-relevant causal quantities \cite{vanderweele2013causal, heiler2024heterogeneous}.

This paper makes three primary contributions. First, we organize common reasons for apparent effect heterogeneity in state-policy studies, including population composition, contextual environment, policy co-occurrence, policy grouping, and policy implementation, and explain why standard causal estimands cannot distinguish among all these factors separately. Second, we distinguish several causal estimands that may be conflated in applied research: the state-specific treatment effect (ITE), the conditional average treatment effect (CATE), and the controlled direct effect (CDE), and we discuss how treatment coarsening further complicates the definition of these quantities. Third, motivated by the conceptual challenges and data limitations in this setting, we propose bounding state-specific treatment effects as an alternative inferential goal. Building on \cite{manski2018right,rambachan2023more}, we develop one bounding approach within a difference-in-differences framework, discuss how sensitivity parameters can be informed by pre-treatment data, and illustrate the performance of this method in simulations. Finally, we apply this method to estimate state-specific effect bounds for the Affordable Care Act Medicaid expansion on high-volume buprenorphine prescribing.

We begin by discussing why state policy effects may vary across states at all, both cross-sectionally and over time. 

\section{Potential factors driving effect heterogeneity}\label{sec:2}

We propose five broad sources of effect heterogeneity that may apply to any study of state policy changes that operate cross-sectionally at a specific time-period. These categories are not mutually exclusive, and some proposed categories may be reframed as another. Moreover, our proposed categories fall into two different conceptual groups: the first three discuss heterogeneity under a well-defined intervention, while the second relate to heterogeneity induced by grouping similar policies as a single intervention. We then discuss factors that drive heterogeneity across time, again separating genuine time-variation from analyst-induced apparent heterogeneity. Throughout, we assume no interference across states: that is, each state's outcome depends only on its own policy, so that a state-specific effect is well defined. Were this to fail, for instance through cross-border spillovers, differential exposure across states would itself generate apparent heterogeneity, both cross-sectionally and over time, though further exploring this is beyond the scope of this paper.

\subsection{Cross-sectional effect heterogeneity}

\subsubsection{Variation in state population composition}

When individual‑level traits exist that modify a policy’s impact and that are distributed unevenly across states, aggregate state-level policy effect differences may emerge.

\textbf{Example:}
\begin{itemize}
    \item Suppose a naloxone education policy more strongly influences younger adults’ behavior. For an analysis that averages over the entire population, states with younger populations may appear to experience larger effects even though the policy mechanism operates through individual-level responses.
\end{itemize}

\subsubsection{Variation in state contextual factors}
Differing physical, social, and economic contexts across states may alter a policy's efficacy. Crucially, these are features of a person's \emph{environment}, not of the person themselves, and so they may modify a policy's effect even holding individual characteristics fixed.

\textbf{Examples:}
\begin{itemize}
    \item Suppose a naloxone education policy works partly through social diffusion: an individual who learns about naloxone shares that knowledge with peers. The policy’s effect may then be more pronounced in a more socially connected community, where information propagates more readily across social networks. Here the policy's effect on a given person depends not only on their own characteristics but on the social environment around them---a contextual effect distinct from the population-composition channel discussed above.
    \item Good Samaritan Laws may work better in densely populated states where bystanders are more likely to witness overdoses.
    \item Variability in the illicit drug market, for instance prevalence of fentanyl or other adulterants, may change the apparent effectiveness of overdose‑prevention policies.
\end{itemize}

\begin{remark}
    When only state-level data are available, as we assume throughout, individual-level effect modification cannot be separately identified from contextual effect modification operating through the same aggregated characteristic. Unless one is willing to assume, in a given application, that either the individual-level or contextual channel is absent, a state-level heterogeneity analysis captures their combined contribution. In this sense, the aggregate characteristic may be useful for describing effect heterogeneity, but state-level data alone do not identify the mechanism through which that heterogeneity arises.

    We nevertheless retain the distinction because the two mechanisms are conceptually different: a state-level effect modifier can be viewed as combining an \emph{individual-response} channel, that is, how an individual responds to the policy as a function of their own characteristics, and a \emph{contextual} channel, or how an individual responds as a function of their environment, such as that same characteristic aggregated across the state. State-level data may identify both effects together, not each component (or their interaction). Nevertheless, naming the two channels clarifies what an estimated association may, and may not, reveal about the mechanism driving heterogeneity.\footnote{Formally, for individuals $i$ in state $j$, if $Y_{ij}(1)-Y_{ij}(0) = \beta X_{ij} + \gamma \bar X_j^{-i}$, where $\bar X_j^{-i}$ is the state-average without individual $i$, then the state-level effect is $\bar Y_j(1)-\bar Y_j(0) = (\beta+\gamma)\bar X_j$. In other words, State-level data may identify $\beta+\gamma$ but neither component separately (though, of course, with individual-level data we could feasibly separately estimate these components). This model could readily be extended to include interactions or non-linear terms, but is kept simple here for illustration.}
\end{remark}

The preceding sources describe genuine heterogeneity for a well-defined policy. The following two sources instead reflect heterogeneity induced by analyst choices about treatment definition that may manufacture apparent heterogeneity, even when the underlying effects are constant.

\subsubsection{Policy climate and co-occurring policies}
State policymakers rarely change one law at a time. The effect of a single policy in the presence of another policy may differ \citep{schuler2020state}, leading to effect heterogeneity when different states have different policy environments.\footnote{State policy climate could be considered part of a state's contextual environment. However, in the context of conducting policy evaluation studies, state policy climate is a key factor worth explicitly highlighting.} 

\textbf{Examples:}
\begin{itemize}
    \item Good Samaritan Laws may amplify the benefits of Naloxone Access Laws (NALs) by reducing the fear of calling emergency services.
    \item During the COVID‑19 pandemic, restrictions on in‑person appointments may have dampened the impact of continuing-education mandates intended to improve clinicians’ knowledge of addiction treatment, as well as other policies designed to expand treatment access.
\end{itemize}

\subsubsection{Policy grouping}
In contrast to randomized experiments, where interventions are ideally operationalized as a uniform set of instructions or a manual \cite{almirall2016adaptive}, uniformity of policy interventions in observational settings, including state-level policy settings, is even less common. Instead, grouping related policies together as a single intervention is a common practice, and is sometimes explicitly recommended \citep{seewald2024target}. While this may increase the statistical power of an analysis when the grouped policies have similar effects, it may reduce power when they do not, due to additional noise induced by underlying effect differences.\footnote{Grouping continuous treatment variables into categories is also sometimes recommended in the difference-in-differences literature as a way to increase power \cite{callaway2024difference,hettinger2025multiply,zhang2025continuous}. Again these groupings, in either a difference-in-differences or more general setting, may suffer from the same problems if treatment effects in fact vary substantially within categories.} Moreover, in this latter case, policy grouping may mask meaningful legal differences that correlate with policy efficacy, leading to apparent heterogeneity with regard to the aggregated, or coarsened, policy indicator \citep{heiler2022heterogeneity}.

\textbf{Examples:}
\begin{itemize} 
    \item Binary indicators for Prescription Drug Monitoring Program (PDMP) adoption conceal meaningful differences between so-called mandatory‑use programs, which require prescribers and/or pharmacists to check the PDMP prior to prescribing or dispensing a controlled substance such as opioids, and voluntary programs, which require PDMP registration but not their use \citep{schuler2021methodological}. Using a binary indicator to group these PDMP policies may obscure true variation in PDMP effectiveness, as studies have suggested that mandatory use PDMPs have a stronger impact on opioid prescribing than voluntary use PDMPs. 
    \item NALs include diverse provisions, including standing orders, third‑party prescribing, co‑prescribing mandates, and community distribution, that vary widely in scope across states and may also result in varying efficacy of the overall policy.
\end{itemize}

\subsubsection{State differences in policy implementation and enforcement}
Even identical statutes may differ in execution. Financing, administrative capacity, the timing of implementation, and enforcement intensity influence realized effects \citep{mcginty2024scaling}. This can be viewed as a form of the policy-grouping problem discussed above: the written statute may be the same, but the policy as delivered differs across states. In this sense, implementation is part of the intervention ``manual'' that would ideally define the policy change, alongside the formal legal provisions themselves.

\textbf{Examples:}
\begin{itemize}
    \item PDMPs vary in integration with electronic health records and in proactive monitoring practices.
    \item Regulatory environments for legal recreational cannabis vary widely across states in terms of local regulations for licensing, retail density, product potency, and taxation; these factors impact both access and consumption.
\end{itemize}

We again emphasize that the first three sets of factors -- demographic and contextual factors and the policy environment -- pertain to the efficacy of a distinct policy and how that may change depending on other features of the state environment. The second two factors pertain to an analyst's decisions about how to define a policy. In fact, as we will later show formally, the underlying policies may have constant effects across states, but the analyst's decision to group them as a distinct intervention will result in apparent heterogeneity. Regardless, for a given analysis, all aforementioned sources of variation may operate simultaneously and interact. 

\subsection{Time-varying effect heterogeneity}\label{ssec:timevarying}

So far we have considered heterogeneity across states at a single post-treatment time-period. When effects are instead measured over multiple periods, an additional layer of heterogeneity arises: the realized effect of a policy may differ depending on \emph{when} it is measured and \emph{how long} it has been in force. As in the cross-sectional case, it is useful to separate sources that reflect genuine variation in the effect of a well-defined policy from those that reflect the analyst's choices about how effects are aggregated. We discuss each in turn, and then highlight a conceptual challenge specific to staggered adoption.

\subsubsection{Dynamic effects and calendar-time variation in realized effects}

A key source of time-variation is the dependence of a policy's effect on the duration it has been in effect, often called dynamic effects. This may arise as implementers learn to administer a policy more effectively, as individuals and institutions gradually adjust their behavior, or simply because effects accumulate, decay, or otherwise change with exposure length \citep{schuler2020state}.

\textbf{Examples:}
\begin{itemize}
    \item As people become more aware of and comfortable with naloxone following the adoption of a NAL, the law's efficacy may grow with time since adoption.
    \item Some have hypothesized that NALs may increase opioid misuse insofar as they make opioids appear safer \cite{tas2019should}. If true, such behavioral responses may increase as a function of time since NAL adoption.
\end{itemize}

Second, the realized effect of a policy may differ across calendar periods because the conditions under which the policy operates change over time. State demographic and contextual composition may shift, complementary policies may be enacted or repealed, and administrations may change how policies are implemented and enforced. Consequently, even if the effect of a policy were stable conditional on these factors, the realized state-level effect could depend on the calendar period in which it is measured.

\textbf{Examples:}
\begin{itemize}
    \item Suppose again that a naloxone education policy influences only younger adults' behavior, and that young people are disproportionately moving to a particular state over time. We may then expect larger effects in later periods.
    \item A change in administration may alter the implementation of a state's social services program, shifting the health care context in which any policy operates and thereby changing its effect across time periods.
\end{itemize}

\begin{remark}[A conceptual challenge specific to staggered adoption]\label{rmk:sa}
    When all treated states adopt at a single time $T$, calendar time and exposure duration are collinear by construction: there is a single post-treatment horizon, and, without strong assumptions, the two cannot be told apart simply because they do not vary separately. Under staggered adoption, by contrast, states adopt at different times, so for a state adopting at $T_i$ and observed at $t > T_i$ the realized effect reflects both the calendar period $t$ (through time-varying factors) and the exposure duration $t - T_i$ (through dynamic effects). These two channels now vary separately across cohorts, but they are generally confounded: without further assumptions, a state's effect at $t$ cannot be decomposed into ``the policy operates differently in this calendar period'' versus ``the policy has been in effect longer.'' This mirrors the cross-sectional non-identifiability of individual- versus contextual-level effect modification discussed in Section \ref{sec:2}: in both cases state-level data are typically thought to identify a combination of conceptually distinct sources, not their separate contributions.
\end{remark}

\subsubsection{Analyst-induced apparent time-variation: aggregation across cohorts}

Just as treatment coarsening and differential implementation can produce apparent cross-sectional heterogeneity (Section \ref{sec:2}), analytic choices of how to aggregate effects over time and across adoption cohorts can generate apparent time-varying dynamics that do not reflect change in the underlying policy effects.\footnote{Beyond these estimand-definition issues, particular estimators may 
introduce apparent time-variation as an artifact of bias. For example, two-way fixed-effects and event-study specifications under staggered adoption can produce spurious apparent dynamic effects when treatment effects are heterogeneous, because already-treated units serve as contaminated controls for later adopters 
\citep{goodman2021difference, gardner2022two}. However, these are estimation rather than identification concerns, and can be addressed through alternative estimation strategies.}

\textbf{Examples:}
\begin{itemize}
    \item States that adopted NALs early (e.g., the first wave around 2013–-2015) may differ systematically from later adopters in baseline opioid burden and treatment infrastructure. An estimand that averages the effect of NAL adoption across this changing set of adopters may then display apparent time ``dynamics'' that reflect who is adopting in each period rather than any change in a given state's response over time.
\end{itemize}

\section{Causal estimands quantifying heterogeneity}\label{sec:1}

This section formalizes several causal estimands that are commonly invoked in state-policy heterogeneity analyses but are often conflated in practice: the state-specific treatment effect (ITE), the conditional average treatment effect (CATE), and the controlled direct effect (CDE). We use these definitions to clarify what kinds of policy questions each estimand can and cannot answer, and to explain how treatment coarsening further complicates their interpretation. For conceptual clarity, we focus on effects at a single post-treatment time-period.

\subsection{Setup}

We observe $i = 1, \dots, N$ states over $t = 1, \dots, T$ time periods, where treatment occurs at time $T$. Let $M_i \in \{0, 1, \dots, k\}$, where $M_i = 0$ represents no policy adoption, and $M_i = m > 0$ represents the adoption of a specific policy or implementation. We define the coarsened treatment indicator $A_i = \I(M_i > 0)$. Let $Y_{it}$ denote the observed outcome and $Y_i(m)$ denote the potential outcome under treatment $M_i = m$ evaluated at time $T$ (we suppress the dependence on time for simplicity). We assume SUTVA -- in particular, the no-interference-between-states assumption noted in Section \ref{sec:2}, so that $Y_i = Y_i(M_i)$ depends only on state $i$'s own policy. We also assume no anticipatory effects, so that $Y_{it}(m) = Y_{it}(0) = Y_{it}$ for all $m$ when $t < T$. Finally, let $S_{it} =(Z_{it}, X_{it})$ denote observed state-level covariates, which may summarize individual-level variables, contextual features, or other aspects of the policy environment. We assume $S_{it}$ is unaffected by policy adoption ($S_{it} = S_{it}(m), \forall m$). We partition $S_{it}$ simply to recognize that we may be interested in heterogeneity with respect to some subset of the covariates $X_{it}$, rather than all available covariates $S_{it}$. We therefore focus our discussion on heterogeneity with respect to $X_{it}$.

Finally, we let $\P_N(O)$ denote the empirical average of some variable $O$ over the $N$ states, ($N^{-1}\sum_{i=1}^NO_i$), and $\Pr_N$ to denote the probability of some discrete variable taking on some value over the $N$ states $(\Pr_N(O_i = o) = N^\inv \sum_{i=1}^N\I(O_i = o)$ for a discrete variable $O$).

\subsection{Causal estimands: well-defined treatment}

We first define causal quantities, or estimands, relevant to heterogeneity analyses, noting that these are not the only estimands one may consider (see, e.g., \cite{hettinger2025causal}). The first is the state-specific treatment effect (ITE).\footnote{ITE conventionally stands for individual-treatment effect, where the individual is the unit. We use the term ``state-specific'' to clarify that we are not talking about individuals in a population at the sub-state level.}

\begin{definition}[State-specific treatment effect]
    For state $i$ and policy version $m$,
\begin{align*}
    \psi_{i,m} = Y_i(m) -Y_i(0).
\end{align*}    
\end{definition}

\noindent The ITE represents the effect of implementing policy $M = m$ versus no policy adoption for a specific state at the observed covariate level $S_i$. $\psi_{i,m}$ may differ across states for the first three reasons noted in Section \ref{sec:2}, including individual-characteristics, contextual factors, or the state policy environment. However, by definition we preclude heterogeneity due to treatment coarsening or differential implementation, since $\psi_{i,m}$ is by definition a function of a fixed and well-defined policy. 

Even if we knew $\psi_{i,m}$ for all units, we would not know what factors either cause or associate with variation in $\psi_{i,m}$ across units. The next two estimands attempt to quantify this. First, we consider the conditional average treatment effect (CATE).

\begin{definition}[Conditional average treatment effect]
    For policy $m$ and covariate value $x$,
    \begin{align*}
        \tilde\psi_m(x)= \E[Y(m)-Y(0)\mid X=x].
    \end{align*}
\end{definition}

The value $\tilde\psi_m(x)$ is a pointwise conditional causal contrast that describes the average effect of policy $m$ among units with $X=x$. Effect heterogeneity with respect to $X$, however, is characterized by the shape of the CATE function across values of $x$. Comparisons along this curve reveal how treatment effects are \emph{associated} with $X$, but do not by themselves imply that intervening on $X$ would change the treatment effect.

We may alternatively define the controlled direct effect (CDE) \cite{pearl2022direct},

\begin{definition}[Controlled direct effect]
    For policy $m$ and covariate value $x$,
    \begin{align*}
        \psi_m(x) = \E[Y(m,x)-Y(0,x)].
    \end{align*}
\end{definition}

Analogously, $\psi_m(x)$ is a pointwise causal contrast that describes the average effect of policy $m$ under an intervention that sets $X=x$. The shape of the CDE function across values of $x$ describes how the CDE varies across possible interventions on $X$. Unlike the CATE curve, variation in the CDE curve has a \emph{causal} interpretation with respect to $X$.\footnote{We could alternatively consider the conditional estimands $\psi_m(x,z) = \E[Y(m,x)-Y(0,x)\mid Z = z]$ or $\tilde\psi_m(x,z) = \E[Y(m)-Y(0) \mid X = x, Z = z]$ with no additional complications in the remaining discussion.} In what follows, we distinguish between evaluating an estimand at a fixed value of $x$ and comparing the estimand across values of $x$. The former gives a pointwise causal contrast, while the latter describes the shape of the effect curve over the covariate space.

\begin{remark}[Longitudinal extensions]
    We may also define analogues to these estimands in a more general longitudinal setting. Letting $\bar M_{it}$ be the observed sequence of policies from time $T_0 + 1,\dots, t$ that takes any value $\bar m_t$, we may define the potential outcome $Y(\bar m_t)$ and $Y(\bar 0_t)$ to define causal contrasts. In cases where $X_{it}$ is also time-varying, we may use the same notation. The discussion below will also apply in this setting.
\end{remark}

The ITE, CATE, and CDE are thus conceptually distinct and need not coincide. To illustrate, consider the case where the intervention is NALs and PDMPs is the covariate. The state-specific effect asks: what is the effect of a particular kind of NAL for a specific state? The CATE asks: what is the average effect of a particular kind of NAL among states that have PDMPs? The CDE asks, what would the average effect of a particular kind of NAL be if we could force all states to adopt PDMPs? Comparisons of CDEs at different combinations of $x, m$ could justify causal claims such as, ``A particular kind of NAL works better on average in combination with PDMPs'' (e.g., for binary $X$, $\psi_m(1)-\psi_m(0)$). By contrast, a similar comparison of CATEs ($\tilde\psi_m(1) - \tilde\psi_m(0)$), only supports the associational statement ``The average effect of a particular kind of NAL is larger among states with PDMPs than among states without PDMPs.'' In this second case the larger effect may or may not be due to PDMPs or some other factor associated with states having PDMPs. Finally, the ITE need not equal the CDE or the CATE, even in sign, unless there is no unobserved variation in treatment effects beyond that explained by the particular NAL and PDMPs. However, this is unlikely ever to hold, which highlights an important point: while some refer to the CATE as an ``individualized treatment effect,'' such effects are in truth an average that may differ even in sign from the true ITE \cite{vegetabile2021distinction}.

In brief, the CATE function blends causal and associative components, and may appear to vary across $X$ even when the quantity $Y_i(m,x)-Y_i(0,x)$ is constant.\footnote{This differs from the previously defined ITE, which evaluates the outcome at the realized values of $X_i$, so that $\psi_{i,m} = Y_i(m, X_i) - Y_i(0, X_i)$, and therefore $Y_i(m) = \sum_x Y_i(m, x)\I(X_i = x)$.} In state‑policy settings, this kind of confounding is common: unmeasured demographic or political contexts often correlate with other measured effect modifiers such as income, rurality, or various state policy indicators. Consequently, we argue that the CDE function, which reveals strictly causal relationships, typically has greater policy relevance than the CATE \cite{hettinger2025causal}. The ITE, on the other hand, also gives policy relevant information -- what was or would be the causal effect for a specific state -- however, it does not suggest what factors may modify the treatment effect, for which even the CATE provides some, albeit associational, evidence. Finally, we can rarely generalize from the CATE or the CDE to the ITE, in the same way we cannot generalize averages to make inferences about states or individuals \cite{vegetabile2021distinction}. 

We conclude by discussing a more abstract but commonly estimated quantity: a \textit{linear projection} of the CATE. Because data in state-policy analyses are limited, it is common to specify a working model of the CATE function, such as $\tilde\psi_m(x) = \tilde\delta_m + \tilde\beta_mx$, and estimate this model using OLS. However, like all parametric models, these models are likely to be misspecified -- that is, the true CATE function is not likely to be a simple line with respect to $x$, but some more complex function. Fortunately, we may interpret corresponding OLS parameter estimates as the best linear projection of the true underlying CATE: heuristically, this can be thought of as a line that best approximates the CATE \cite{angrist2009mostly}. This projection may or may not be a desirable quantity to estimate, depending on how well the specified line fits the true function. While we may not truly be interested in the projection parameters, from a practical standpoint, unless $X$ represents a small number of binary indicators, projections are often statistically easier to estimate than the true CATE \cite{angrist2009mostly, buja2019models, rubinstein2023heterogeneous}.

\begin{definition}[OLS projection of the CATE]\label{def:proj}
    We define the least-squares linear projection of the CATE as,
    \begin{align*}
       x^\top\beta_m, \qquad \beta_m=\arg \min_{b} \E\bigl[(\tilde\psi_m(X) - X^\top b)^2\bigr].
    \end{align*}
\end{definition}

\noindent Any of the aforementioned factors -- individual, contextual, and policy differences -- may drive apparent heterogeneity across either the ITE, CATE, CDE, or the linear projection of the CATE. However, we have so far not considered how treatment coarsening or differential implementation may affect the interpretation of an analysis.

\begin{remark}
    The use of the expectation function $\E$ in the definitions above suggests that the units are drawn from some larger population. In this case, these estimands are population effects. However, we may also define analogous estimands solely with respect to the observed sample, replacing expectations with empirical averages under $\P_N$. In this case, these quantities are sample effects. For state-level policy analyses, this latter concept is often more natural, since there is no literal super-population of states. It may nevertheless be useful to conceive of such a super-population, both to define causal quantities and to conduct statistical inference. For example, inference for population quantities can be conservative for the corresponding sample analogues when only treatment assignment is viewed as random \citep{imbens2004nonparametric}. Moreover, some sample analogues may not be well-defined for all covariate values. For instance, the sample CATE is not defined at unobserved values of $X$, which is especially relevant when $X$ is continuous. For simplicity, the remaining discussion uses notation that suggests a super-population framework, but we emphasize that sample averages are often the more natural conceptual targets in state-level policy analyses.
\end{remark}

\subsection{Causal estimands: coarsened treatment}\label{ssec:coarsening}

Distinct policy versions or implementations are often grouped into a single treatment indicator $A$ in many applications (see, e.g., \cite{rubinstein2023balancing}). For treated states, we may take $Y(A = 1) = Y(M)$, and define the ITE, CDE, and CATE among treated states as before. However, without imposing further structure on this problem, $Y(A=1)$, and therefore the CATE and CDE are not well-defined among non-treated states. In Appendix \ref{appsec:coarsening}, we consider one solution to this problem, largely following \cite{vanderweele2013causal}. In short, we view $(A, M)$ as a two-stage randomization process. First, $A$ is randomized to determine whether $M = 0$ or $M > 0$; second, $M$ is randomized for $m > 0$ when $A = 1$ \citep{vanderweele2013causal, heiler2022heterogeneity}. We further assume that $A$ has no effect on $Y$ except via $M$, so that we may then define the potential outcome $M(A = 1)$ for all units, and set $Y(A = 1) = Y(1,M(A = 1)) = Y(M(A=1))$.\footnote{In many applications the ``control'' regime is also not literally unique, so that $Y(A = 0)$ is also not well-defined. For example, untreated states may remain under different status quo policies (i.e., there may be multiple versions of $M=0$). When this occurs, one can either expand $M$ to index those status quo versions as additional values, or interpret $Y(A=0)$ as a versioned control potential outcome $Y(M(A=0))$, as in \cite{vanderweele2013causal}. However, for simplicity, we proceed by assuming a single well-defined control condition, so that $Y(A = 0)$ is well-defined for all states. The key conceptual point is that coarsening can create mixtures of policy versions on the treated side -- and potentially on the control side as well.}

We may define coarsened versions of the ITE, CATE, and CDE (and their projections), which we denote as $\psi_{i,M(1)}, \tilde\psi_A(x)$, and $\psi_A(x)$, respectively, as follows: 

\begin{definition}[Coarsened ITE]
    \begin{align*}
        \tilde\psi_{i, M_i(1)} &= Y_i(M_i(A_i=1)) - Y_i(A_i=0).
    \end{align*}
\end{definition}

\noindent The coarsened ITE is simply the ITE $\psi_{i,m}$ evaluated at the value of $m = M(1)$, motivating the notation $\psi_{i, M_i(1)}$. We next define the coarsened CATE and CDE as,

\begin{definition}[Coarsened CATE]
    \begin{align*}
        \tilde\psi_{A}(x) &= \E[Y(A=1) - Y(A = 0) \mid X = x] \\
        &= \sum_{m > 0}\E[Y(m) - Y(0) \mid M(1) = m, X=x]\Pr(M(1) = m\mid x),
    \end{align*}
\end{definition}
\noindent and
\begin{definition}[Coarsened CDE]
    \begin{align*}
        \psi_{A}(x) &= \E[Y(A=1, x) - Y(A = 0, x) ] \\
        &= \sum_{m > 0}\E[Y(m, x) - Y(0, x) \mid M(1) = m]\Pr(M(1) = m\mid x).
    \end{align*}
\end{definition}

These definitions show that the coarsened CDE are weighted averages of  $m$-specific CATEs and CDEs that further condition on the stratum where $M(1) = m$, and where the weights are given by the conditional distribution of $M(1)$ given $X$. As noted in Appendix \ref{appsec:coarsening}, when $A$ is randomized with respect to $M(1)$ given $X$, this is equivalent to the conditional distribution of $M$ among the treated states.

While we may mathematically define coarsened versions of the ITE, CATE, and CDE, the interpretation of these quantities remains challenging. For example, any observed heterogeneity with respect to $A$ may simply reflect heterogeneity arising from different versions of treatment. We formalize this in Lemma \ref{lemma:1} below.

\begin{lemma}\label{lemma:1}
    Assume that $Y_i(m)-Y_i(0) = C_m$ for all units and that $k > 1$. Unless (i) $C_m = C$ for all $m$, or (ii) $M(1) = m$ for all units, then $\psi_{i,M(1)}$ will vary across units.
\end{lemma}

\noindent Lemma \ref{lemma:1} tells us the state-specific effects will appear to vary even when the underlying effects are constant. While this only directly pertains to the ITE, the implications extend to the CATE and the CDE. For example, consider the case where the effects of must-access and standing order NALs are constant holding PDMP status fixed. Further assume that must-access NALs are always more effective than standing order NALs. In this case, if must-access NALs were more prevalent in states with PDMPs, and standing order NALs were more prevalent in states without PDMPs, then the corresponding CDE (and possibly the CATE) defined with respect to ``NALs'' would show positive effect modification with respect to PDMPs, even when there is none. 

Finally, we have not discussed linear projections of the CATE. While it is easy to apply definition \ref{def:proj} to the coarsened CATE, the interpretation of an already conceptually challenging quantity becomes even worse. In brief, treatment coarsening further complicates the interpretation of any heterogeneity analysis since the estimands now average over different causal effects with no clear policy implication.

Figure \ref{fig:illustration} illustrates the estimands described above using a single simulated dataset generated from the general model specified in Appendix \ref{appsec:model} and with specific parameters detailed in Appendix \ref{appsec:sim}. In brief, we draw effect modifiers $(X,U)$ jointly from a mean-zero bivariate normal distribution, where $X$ is an observed effect modifier of interest and $U$ is an unobserved effect modifier correlated with $X$, with correlation equal to 0.25. We then assign treatment using the two-stage framework described in Section \ref{ssec:coarsening}: first, we randomly assign a coarsened treatment indicator $A$ as a probit function of $X$; second, among the treated units ($A=1$), we randomly assign $M\in\{1,2\}$ as an additional probit function of $X$.

We then generate the potential outcomes using a simple linear model in which both treatment versions have effects that vary with $X$ and $U$. In the parameterization used here, the $X$-slope of the version-specific effect is positive for $M=1$ and negative for $M=2$, so the corresponding CDEs $\psi_m(x)$ have opposite signed relationships with $x$. Because $U$ is correlated with $X$, the corresponding CATEs $\tilde\psi_m(x)$ differ from the CDEs within each treatment version, illustrating how confounding between $X$ and unobserved effect modifiers may distort the CATE relative to the CDE.\footnote{In this illustration the conditional effect functions are linear in $x$, so their linear projections coincide with the true functions; this need not hold in general.} Finally, we plot the state-level effects (ITEs) for the realized treatment version as semi-transparent points to emphasize that state-specific effects include idiosyncratic variation beyond the CATE/CDE averages.

\begin{figure}
 \begin{center}
     \includegraphics[scale=0.5]{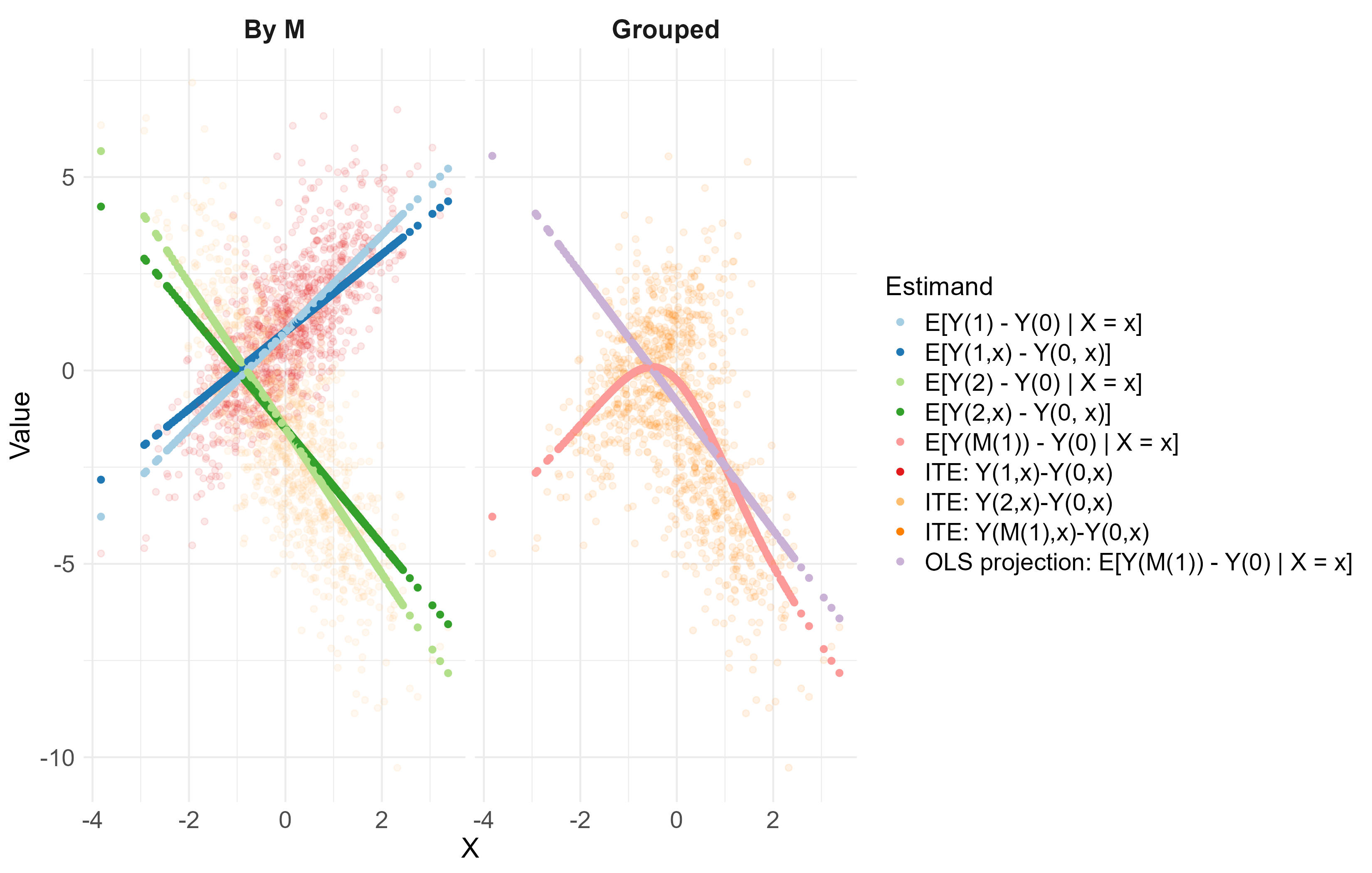}
\end{center}
\caption{Causal estimands in a single simulated dataset. Left panels show the version-specific CDEs $\psi_m(x)$ and CATEs $\tilde\psi_m(x)$ for $m=1,2$, together with the realized state-level ITEs (semi-transparent points). The right panel collapses treatment into the coarsened indicator $A$, showing the coarsened CATE, its OLS linear projection, and the ITEs. The projection's negative slope masks the positive $x$-relationship visible when treatment versions are modeled separately.}\label{fig:illustration}
\end{figure}

On the right-hand panel, we instead collapse treatment into the single indicator $A$. The resulting coarsened CATE differs markedly from the left-hand panels: it is roughly parabolic, increasing for $x<0$ and decreasing for $x>0$. This shape reflects the fact that the mixture over treatment versions changes with $x$: units with $x<0$ are more likely to receive $M=1$, whereas units with $x>0$ are more likely to receive $M=2$.\footnote{In the simulation we also have that $M(1)$ is independent of $A$, so that the coarsened CATEs/CDEs is simply a weighted average of the $m$-specific CATEs/CDEs.} Meanwhile, the OLS target of inference -- the linear projection of the coarsened CATE -- captures the region where $x>0$ somewhat well but misses the overall shape of the true function, and its negative slope masks the positive relationship that would be apparent if the treatment versions were modeled separately. Finally, the ITEs, defined at the realized value of $M(1)$, deviate from the coarsened CATE due to additional idiosyncratic variation.

This is a deliberately adverse scenario: for a given application, it may be unlikely that different treatment types would have opposite signed relationships with an effect modifier. Nevertheless, this scenario highlights that a common target of inference -- the OLS projection of the coarsened CATE -- may yield little insight about the causal relationships between a policy change and an effect modifier even in a simple setting. Appendix \ref{appsec:model} provides a more general but simple mathematical illustration showing these results using linear models, again illustrating why the CATE, CDE, and ITE diverge. 

\subsection{Causal identification and estimation}

We have shown that the quantities we most often estimate -- the CATE, coarsened estimands, and their linear projections -- lack clear interpretations, whereas the ITE and CDE, absent treatment coarsening, are generally the most policy-relevant. A natural response is to avoid coarsened estimands and estimate the CDE or ITE directly within well-defined treatment levels. Data realities in state-policy settings, however, make this difficult on two fronts. First, \emph{identification}: the CDE and ITE require stronger and often implausible assumptions than the CATE, and because policy changes are rarely randomized, identification further demands adjusting for all confounders through models we cannot credibly specify. Second, \emph{estimation}: with at most fifty states, even an identified estimand is typically underpowered, and uncertainty quantification is unreliable when few units are treated. We take each in turn.

To illustrate the identification challenges, consider the case where a policy is completely randomized, so that $Y(m) \perp M$ for all $m$. In this case, the CATE is identifiable as $\tilde\psi_m(x)=\E[Y\mid M = m, X = x] - \E[Y \mid M = 0, X = x]$. However, this does not equal the CDE unless, for example, $Y(m,x) \perp (M, X)$, and randomization of $M$ does not imply randomization of $(M, X)$. Similarly, under a
difference-in-differences framework, identification of $\E[Y(m,x)]$ generally proceeds from a conditional independence or mean independence assumption with respect to $X$ and $Y(m,x)$ in addition to the standard parallel-trends assumption \citep{hettinger2025causal}.\footnote{We emphasize that, when pre-treatment data are available for all units, parallel-trends allows identification of $\E[Y(0) \mid A
= 1]$ under relatively weak assumptions, since this uses the fact that the pre-treatment potential outcomes absent treatment are observed for all units. However, heterogeneity analyses that generalize to all units in a sample require stronger assumptions, as the outcomes under treatment are not observable in the pre-treatment period.} Finally, the ITE $\psi_{i,m}$ is almost never identified, since we only observe $Y(M)$ for any unit; were we willing to assume effects are constant given the CATE or CDE, we could identify it using $\E[Y \mid M = m, X] - \E[Y \mid M = 0, X]$, but this is unlikely to hold. Thus the assumptions required to point identify the CATE are the weakest, those for the CDE stronger, and those for the ITE the strongest. Table \ref{tab:4} summarizes the combinations of assumptions under which these estimands are identified and when they coincide.

\begin{table}[ht]
\centering
\caption{Progression of identified estimands and equality conditions as assumptions strengthen.}\label{tab:4}
\begin{tabular}{p{5.4cm} l ccc}
\toprule
\multicolumn{2}{c}{} & \multicolumn{3}{c}{\textbf{Equalities that hold}}\\
\cmidrule(l){3-5}
\textbf{Assumptions (cumulative)} & \textbf{Identified estimands} & CATE{=}CDE & CATE{=}ITE & CDE{=}ITE\\
\midrule
1. Identified treatment variable                       & CATE*              & No  & No  & No  \\
2. \hspace{0.6em}+ well-defined treatment              & CATE               & No  & No  & No  \\
3. \hspace{0.6em}+ effect modifier correctly controlled & CATE, CDE          & Yes & No  & No  \\
4. \hspace{0.6em}+ constant conditional effects        & CATE, CDE, ITE     & Yes & Yes & Yes \\
\bottomrule
\end{tabular}
\footnote{*CATE* refers to the fact that the CATE is a weighted average of treatment-specific CATEs when the treatment is not well-defined. We use the asterisk to emphasize that this is not technically a CATE with respect to a well-defined policy intervention.}
\end{table}

Identification is further complicated by the fact that state policies are rarely completely randomly assigned. Whatever identifying conditions one assumes, in almost all applied settings we must also adjust for confounders $S = (Z, X)$ -- both the effect modifiers of interest $X$ and other confounders $Z$, which may be binary, categorical, or continuous. Obtaining unbiased or consistent estimates of an identified parameter then requires either (1) correctly specifying the relationships among the confounders, treatment, and outcome; or (2) using doubly-robust estimators together with flexible machine learning to avoid parametric assumptions \cite{chernozhukov2018double,hill2011bayesian,rubinstein2023heterogeneous}. While doubly-robust estimation is generally preferable, it tends to perform poorly in state-policy settings with so few units. Analysts are thus pushed back toward strong parametric models that posit simple relationships among the observed variables -- models that, as with specifying the effect-modification function, are unlikely to be correctly specified, so that the resulting estimates are biased even for the less interpretable projection parameters.

To illustrate this last concern, suppose we adjust for covariates in the difference-in-differences analysis through a linear, additive specification, and suppose we even include treatment interactions with the effect modifiers $X$ so as to model effect modification correctly. The remaining confounders $Z$ nonetheless enter only additively, which implicitly assumes the treatment effect does not vary with $Z$. When this homogeneity fails -- that is, when some component of $Z$ is itself an effect modifier correlated with treatment assignment -- additive adjustment no longer removes the associated confounding, and the resulting estimates, including the $X$-interaction terms we do model, are biased. More flexible, doubly-robust adjustment relaxes these functional-form restrictions but, as noted, requires more data than these settings typically provide. The distinction between confounders $Z$ and effect modifiers $X$ is therefore consequential: both may need to enter the model flexibly, yet sufficient data rarely exist to do so.

Indeed, regardless of which causal assumptions plausibly hold, sufficient data likely in general do not exist to estimate the identified quantity precisely given the limitation of at most fifty states -- where in some cases only one or two states receive a specific intervention type. In the best-case scenario, there are twenty-five treated and twenty-five control states and a completely randomized policy change. A standard power calculation (80\% power, Type I error rate $0.05$) shows that such an analysis is only powered to detect effects larger than $0.8$ on the Cohen's d scale -- a large effect. Estimating subgroup effects for a single binary covariate would at best leave twelve or thirteen treated and control states, requiring effects larger than roughly $1.2$ on the Cohen's d scale, assuming normally distributed errors. This lack of power may itself motivate analysts to group treatment variables in order to balance the number of treated and control units, which, as we have argued, generally yields a less interpretable estimand.

Finally, this is to say nothing of uncertainty quantification, which is itself challenging when few units are treated \cite{mackinnon2017wild, mackinnon2018wild, mackinnon2020randomization}. With large $T$, some of these problems may be mitigated by assuming independence (or weak dependence) over time within units, but statistical power will generally remain a challenge.

\subsection{CATE limitations}\label{ssec:catelimit}

The previous discussion reveals a central tension in state-policy heterogeneity analyses, together with several related limitations:

\begin{enumerate}
    \item The causal estimands that are generally easiest to identify (CATEs) are often the least aligned with policy-relevant questions, whereas the estimands that are most policy-relevant (ITEs and CDEs) are the hardest to identify, and are not point identified under realistic assumptions.
    \item Treatment coarsening further weakens the connection between any of these estimands and a meaningful, well-defined policy effect.
    \item Even where an estimand is identified, statistical precision is fundamentally limited by at most fifty states.
    \item These data limitations push analysts toward likely misspecified parametric models and toward less interpretable target estimands (e.g., projections and coarsened estimands), trading interpretability for feasibility.
\end{enumerate}

On this final point, data limitations do not, by themselves, justify estimators whose assumptions are unlikely to hold for the target estimand. While practical constraints are real, researchers should explicitly reason about whether the assumptions required for estimating some causal estimand are plausible, and whether the resulting estimates, even if statistically precise, answer policy-relevant questions.

To be clear, there are realistic situations where researchers may understand the limitations of the CATE and still wish to estimate it, or its linear projection, even under treatment coarsening. For example, assuming a plausible identification strategy exists, the CATE may be a reasonable target of inference when we are interested in the CDE and think that the bias of the CATE with respect to the CDE is small, and when the set of plausible confounders is small. Moreover, coarsening is less of a problem when we believe the policy effects are similar across different treatment versions or implementations and have similar relationships with the covariates. Finally, even if the CATE varies nonlinearly in $X$, a linear projection may approximate the relevant portion of the curve well enough. In these cases, a CATE, a coarsened CATE, or even its linear projection, may correctly suggest causal relationships between policies and effect modifiers, even if the targeted estimand is strictly associational. In short, the CATE may be a ``good enough'' and estimable target for many analyses.

On the other hand, when the most policy relevant questions are answered by the ITE or CDE, when we cannot assume that the bias of the CATE or its projection with respect to the CDE is small, or when the required identifying assumptions for any of these quantities are thought to be strong, CATE estimation may be misleading or simply infeasible. For such analyses, we propose to instead shift the inferential focus from point identification to partial identification of the ITE under a framework where nothing is random but some counterfactual quantities are unobserved. We detail this proposal below.

\section{Bounding the ITE}\label{sec:bounding}

We propose an approach to bounding the ITE that only requires specifying uncertainty about the assumptions connecting the unobserved counterfactuals to the observed data. No formal statistical inference is required, and the resulting bounds are directly interpretable as ``how big the effect could plausibly be.'' While the lack of statistical guarantees may seem like a limitation, conceptually, arguably nothing is random in these settings \citep{manski2018right}. This approach therefore sidesteps specifying conceptually strange sources of uncertainty, such as positing sampling from a super-population of states, allowing us to instead posit that all quantities are fixed but possibly unobserved. Finally, the state-specific effects often have more practical utility for policy-makers, and do not require generalizing a state-specific effect from a CATE or CDE. The cost is that the results depend on whether we believe the assumptions required to guarantee that the bounds are valid. We first outline a straightforward difference‑in‑differences implementation that illustrates these ideas, following previous work from \cite{manski2018right, rambachan2023more}, where we do not incorporate covariates and we assume a well-defined treatment variable. We then discuss extensions to accommodate covariates and coarsened treatment variables. However, we note that one could conceivably bound the state-specific effects using any number of strategies, and that this is simply one approach that we believe may be useful and easy to implement in practice.

\subsection{A simple difference-in-differences framework}

To fix ideas, we recall that we observe states over $t = 1, \ldots, T$ time periods, where treatment occurs only at the final period $T$. We begin by defining the trends in each treatment group,

\begin{align*}
    B_m = \left[\sum_{i=1}^N \I(M_i = m)\right]^{-1}\left[\sum_{i=1}^N (Y_{i, T} - Y_{i, T-1})\I(M_i = m)\right], \\
\end{align*}

\noindent and specifying a parallel-trends type assumption that motivates a unit-level difference-in-differences estimator,

\begin{assumption}[Parallel-trends]\label{asmpt:pt}
When $M_i \ne m$, for each state $i$,
\begin{align*}
    \hat Y_{i,T}(m) = Y_{i,T-1} + B_m.
\end{align*}

\noindent where $\hat Y_{i,T}(m) = Y_{i,T}(m)$.
\end{assumption}

\begin{remark}
    We use the trends from $T-1$ to $T$ as the basis for constructing counterfactual outcomes; however, we could alternatively define trends with respect to any weighted average of the pre-period outcomes. We choose to focus on this simple two-adjacent period comparison for both conceptual clarity, and because, as we will discuss, this choice allows the use of the remaining pre-period data to inform possible violations of this assumption.
\end{remark}

For example, $B_0$ represents the average pre–post change among untreated units, which under parallel trends gives the counterfactual trend for each treatment value $m > 0$. This is a finite-sample analogue of the standard parallel-trends assumption. By contrast, when $m > 0$, $B_m$ represents the average pre-post change among treatment group $m$. Under a stronger version of parallel-trends, this gives the counterfactual trend for states where $M_i \ne m$ were it to receive treatment $M_i = m$. We can then apply the unit-level difference-in-differences estimator to obtain the ITE for $M_i \in \{0, m\}$:\footnote{We could also consider causal effects for $\psi_{i,m}$ for all units, but for cases where $M_i \not\in\{0,m\}$ but this would require bounding two unobserved counterfactual quantities, $Y_i(m), Y_i(0)$, and is therefore a more challenging problem that we do not consider.} 

\begin{align}\label{eqn:did}
    \hat\psi^{\mathrm{DiD}}_{i, m} &= \I(M_i = m)(Y_{i,T}-\hat Y_{i, T}(0))+ \I(M_i = 0)(\hat Y_{i, T}(m) - Y_{i, T}), \quad i: M_i \in \{0, m\},    
\end{align}

\noindent When parallel-trends holds, we have that $\hat\psi^{\mathrm{DiD}}_{i,m} = \psi_{i,m}$. This assumption is more naturally motivated for $m=0$ than when $m>0$: since all units are observed under $M=0$ in the pre-treatment period, parallel trends for untreated outcomes corresponds to assuming a continuation of observed pre-treatment trends. By contrast, parallel trends for treated outcomes ($m>0$) requires assuming that untreated states would follow the treated states' trends if they were to receive treatment, a stronger and less directly verifiable assumption.

Regardless, we take the stance that parallel-trends is unlikely to hold for any value of $m$. We therefore instead assume that the errors between the counterfactual estimate $\hat Y_i(m)$ and the true value $Y_i(m)$ are bounded,

\begin{assumption}[Bounded counterfactual error]\label{asmpt:bounded}
    $\lvert\hat Y_{i, T}(m) - Y_{i, T}(m)\rvert \le \tau_i^{m}, \qquad M_i \ne m$.
\end{assumption}

\noindent In truth, the bounded counterfactual error assumption is guaranteed to hold provided the outcomes are finite, as we may simply take $\tau_{i}^m = \tau_i^{m,\star} = \lvert\hat{Y}_{i, T}(m) - Y_{i, T}(m)\rvert$ for assumption \ref{asmpt:bounded} to hold. Any choice $\hat\tau_i^m \ge \tau_i^{m\star}$ may then bound $\psi_{i,m}$ for units $i: M_i \in \{0, m\}$ using,

\begin{align}\label{eqn:6}
\psi_{i,m} \in \left[\hat\psi_{i,m}^{\mathrm{DiD}} - \hat\tau_i^m, \hat\psi_{i,m}^{\mathrm{DiD}} + \hat\tau_i^m \right], \qquad i: M_i \in \{0, m\}.
\end{align}

\noindent The question then becomes how to choose values of $\hat\tau_i^m$ that satisfy \eqref{eqn:6}, and that are ideally as close to $\tau_i^{m,\star}$ as possible. While we can rarely guarantee that a given choice of $\hat\tau_i^m$ will satisfy this, we make some proposals below.

\subsection{Choosing the sensitivity parameters}

When pre-treatment data are available, it is common to use placebo tests in the pre-period to inform the magnitude of possible parallel-trends violations. For example, for treated states, \cite{manski2018right} propose to take $\hat\tau_i^0$ to be the largest pre-treatment violation of the parallel-trends assumption times some constant. To extend this logic to our setting, which seeks to bound the ITE for both treated and untreated states, we can calculate for each pre-treatment time-period $t = 1, ..., T-1$, $\tilde\epsilon_{i,t} = \hat{Y}_{i,t} - Y_{i,t}$ in the pre-treatment period, where $\hat Y_{i,t}$ is the DiD estimate of $Y_{i,t}$. We may then combine these $T-2$ residuals into the vector $\boldsymbol{\tilde\epsilon}_i$,

\begin{align*}
\hat{\tau}_i^{m} = Z_{m} \|\boldsymbol{\tilde\epsilon}_i\|
\end{align*}

\noindent for some norm $\|.\|$. We may then vary $Z_m$ as a sensitivity parameter to evaluate the robustness of the findings to different levels of $Z_m$, where $Z_0 = 0$ implies that assumption (\ref{asmpt:pt}) point identifies $Y_{i}(0)$ for the treated units, and similarly $Z_m = 0$ implies that (\ref{asmpt:pt}) point identifies $Y_{i}(m)$ for the untreated units. 

Other strategies to choose $\hat\tau_i^{m}$ using pre-period data are possible. For example, \cite{rambachan2023more} consider event-study designs, and propose a number of methods, including taking the maximum \textit{difference} in the pre-treatment prediction errors. An analogue to this for the first-differences model we consider above would be to take the final pre-treatment prediction error and add and subtract $Z_m$ times the maximum difference in the pre-treatment errors. However, none of these approaches guarantee $\hat{\tau}_{i}^{m} \ge \tau_i^{m,\star}$. In practice, it may therefore be desirable to try multiple modeling approaches and values of $Z_m$ to see whether the results remain consistent, as we illustrate in Section \ref{sec:application}.

Approaches that use pre-period data to inform $\hat\tau_i^m$ also specify the sensitivity parameter $Z_m$. For analyses bounding both treated and untreated states, we recommend choosing $Z_m > Z_0$ for $m > 0$. Intuitively, the residuals $\boldsymbol{\tilde\epsilon}_{i}$ are more informative about $\tau_i^0$ than $\tau_i^m$, as we are predicting $Y_{it}(0)$ in the pre-treatment period, not $Y_{it}(m)$. We discuss this more formally in Appendix \ref{appsec:bounds}, where we specify a model for the true potential outcomes to illustrate this challenge.

Finally, researchers may wish to vary $Z_m$ to determine a tipping point, that is, the value of $Z_m$ where the resulting bounds would contain zero, thereby explaining away an observed effect. Subject-area expertise or heuristic reasoning may then be used to assess whether this tipping point value is likely to have occurred in practice. Together, these choices allow researchers to express identification uncertainty in a directly interpretable scale tied to observed pre‑period deviations.

Similar approaches to choosing $Z_m$ may also be used to choose $\tau_i^m$ even in cases where one does not have or wish to use pre-period data to inform this choice. That is, researchers may instead simply vary $\hat\tau^m_i$ across a range of values until a tipping point is discovered. As with choosing $Z_m$, this value may informed by subject-area expertise or heuristic reasoning to assess the likelihood that such a value may have occurred.

\subsection{Extensions: covariates and coarsening}\label{ssec:extensionscoarsening}

We now illustrate two practical extensions: incorporating covariates and handling coarsened treatment variables. 

First, we emphasize that there is no need to control for covariates, since we may always adjust the sensitivity parameter $\tau_i^m$ to account for our uncertainty in the validity of our starting assumptions. However, it may be that by controlling for covariates we may obtain a better starting estimate of $Y_{i,T}(m)$, allowing for smaller values $\hat\tau_i^m$ that satisfy assumption \ref{asmpt:bounded}. When $X_i$ are binary or discrete variables, we may simply impute the unobserved counterfactual using the trends among units with the same covariate value. That is,

\begin{align}\label{eqn:2a}
\hat{Y}_{i,T}(m) = Y_{i, T-1} + \underbrace{\left[\sum_{j=1}^N \I[(M_j = m, X_j = X_i)]\right]^{-1}\left[\sum_{j=1}^N (Y_{j, T} - Y_{j, T-1})\I[(M_j = m, X_j = X_i)]\right]}_{B_m(X_i)}.
\end{align}

\noindent This strategy will not work, however, when $X_i$ is continuous, or when the dimension is larger than one or two, given the data limitations in state-policy settings. In this case, we may instead define the bias correction as a linear approximation to the expected change in outcomes given the covariate value among the data with the opposite signed treatment. That is, we may replace $B_m(X_i)$ with

\begin{align}\label{eqn:linapprox}
\hat \E^{lin}_{m}[Y_{i,T}-Y_{i,T-1} \mid X_i, M_i = m],
\end{align}

\noindent where $\hat\E_m^{lin}$ is given by combining OLS parameter estimates from the regression,

\begin{align*}
\arg\min_{\alpha_{0,m},\alpha_{1,m}} \sum_{j: M_j = m}(Y_{j,T}-Y_{j,T-1} - (\alpha_{0, m} + \alpha_{1, m}X_j))^2.    
\end{align*}

\noindent For example, to impute the trends absent treatment for a treated unit, we would use $\hat\E_0^{lin} = \hat\alpha_{0,0} + \hat\alpha_{1,0}X_i$. While other methods to adjust for covariate are possible, regardless of the specific procedure generating $\hat\tau_i^m$ using pre-period data works as before, where whatever chosen procedure is run for each pre-period time period to obtain estimates of the error of this procedure with respect to the observed outcomes.

Treatment coarsening introduces a new challenge for untreated units: instead of predicting $Y(m)$, we must now predict $Y(M(1))$.\footnote{For treated units $Y(M(1)) = Y(M)$ so that we only need to learn $Y(0)$ to obtain $\psi_{i,M(1)}$, and therefore treatment coarsening introduces no particular challenge.} First consider the DiD estimator that uses the coarsened treatment variable $A$. We show in Proposition \ref{prop:lumpite} in Appendix \ref{appsec:didcoarsening}, that under treatment coarsening $\hat\psi_i^{\mathrm{DiD}}$ may be interpreted as a weighted average of $Y(m)$ with weights proportional to the empirical probability of each level of $m$ among the treated states ($\Pr_N(M = m\mid A=1)$). This may or may not be a good starting estimate of $Y(M(1))$, but heuristically, this is a decidedly worse starting estimate of $Y(M(1))$ than the corresponding estimate of $Y(0)$, especially if $Y(m)$ varies substantially over $m$. For example, assume that parallel-trends holds for each value of $m$, so that $\hat Y_{i}(m) = Y_{i}(m)$ for all values $m$. If unit $i$ is treated, then $Y_{i}(0)$ and therefore $\psi_{i,M_i(1)}$ is exactly identified. However, for an untreated unit, parallel-trends does not imply $\sum_{m >0} Y(m)\Pr_N(M = m \mid A = 1) = Y(M(1))$ unless either (1) $Y(m)$ is constant across $m$; (2) $\Pr_N(M = m | A = 1) > 0$ only where $M(1) = m$; or (3) $M(1) = m'$ such that $Y(m') = \sum_{m >0} Y(m)\Pr_N(M = m \mid A = 1)$. None of these cases are likely to hold.

We present three possible approaches to dealing with this problem: (1) choose a large value of the sensitivity parameter ($\tau_i^{M_i(1)}$) to be conservative; (2) assume knowledge that $M_i(1) = m$; (3) take the union of bounds on $Y_i(m)$ for each $m$ to bound $Y(M(1))$. The first approach recognizes the problem but acknowledges that this fact does not invalidate the previously proposed method -- the method is always valid taking $\tau^{M(1),\star}=\lvert Y(M(1)) - \hat Y(M(1))\rvert$. Of course, this value $\tau^{M(1), \star}$ could be made smaller if we had a better starting estimate of $Y(M(1))$, but as with incorporating covariates, this fact does not invalidate the general approach. Alternatively, in the case where the set of policy interventions is small or the interventions are known, we may reasonably assume knowledge of $M(1) = m$ for specific untreated states, and conduct the bounding approach using the DiD estimator with respect to that treatment value. For example, we may think that if some state were to adopt any kind of NAL, it would be a must-access NAL. In that case, one may limit the pool of comparison units among the treated states to where that specific policy was adopted/implemented, estimate $Y(m)$ using equation \ref{eqn:did}, and apply the bounding method as before.\footnote{Of course, in practice even the same laws are often implemented differently, so that most transparent interpretations should interpret the estimate as an average over different treatments; nevertheless, this provides perhaps a more useful estimate of the quantity $Y(M(1))$.} More weakly, perhaps there are some values of $m$ we know the state would not adopt, so that we can at least reduce the pool of comparison states to average over policy versions or implementations that are more likely for a particular state to adopt. Finally, we may remain entirely agnostic about the value of $M(1)$, assuming only that it takes a value $m$ that is observed in the set of treated states, estimate the bounds $Y(m)$ for each value of $m$, and take the union of all the bounds on $\psi_{i,m}$ to yield a bound on $\psi_{i, M_i(1)}$. This may be viewed as a more principled, though also likely more conservative, version of the first approach.  

Overall, while the coarsened DiD ITE estimator of $Y(M(1))-Y(0)$ for untreated units has limitations, simply taking this estimate and applying the bounding method is a reasonable and simple approach. The fact that we may need to choose a larger value of the sensitivity parameter than if we knew $M(1)$ reflects the additional challenge of predicting $Y(M(1))$ relative to $Y(m)$.

\subsection{Extensions: multiple post-treatment time-periods}\label{ssec:staggered}

We also consider extending this method to the setting where we observe multiple post-treatment time periods. Let $T_0$ denote the final pre-treatment period, so that treatment may occur in any period $t > T_0$, and let $\bar M_{it}$ denote the treatment history of a policy from time $T_0+1$ to $t$ for state $i$. We now observe the outcomes $Y_{it}(\bar M_{it})$, and wish to consider a contrast of this observed value with some potential outcome $Y_{it}(\bar m_t)$ under an alternative policy sequence. The choice of contrast operationalizes the distinction between calendar-time and exposure-duration effects raised in Section \ref{ssec:timevarying}.

For any unobserved policy sequence $\bar m_t$ of interest, we may use a difference-in-differences framework to obtain a starting estimate of the counterfactual outcome $Y_{it}(\bar m_t)$ as follows:
\begin{align}\label{eqn:gendid}
\hat{Y}_{it}(\bar m_t) = Y_{i,T_0} + \left[\sum_{j=1}^N \I(\bar M_{jt} = \bar m_t)\right]^{-1}\left[\sum_{j=1}^N (Y_{j,t} - Y_{j,T_0})\I(\bar M_{jt} = \bar m_t)\right], \qquad \forall t > T_0.
\end{align}
The corresponding parallel-trends assumption is that, for states with $\bar M_{it} \ne \bar m_t$, state $i$'s trend from $T_0$ to $t$ would have followed the average trend among states with $\bar M_{jt} = \bar m_t$.

\begin{remark}[Limiting dynamic and time-varying effects]\label{rmk:limitdynamic}
    In some cases, we may believe that dynamic treatment effects last for only a limited number of time periods. Letting $\bar m_t^l = m_t, \dots, m_{t-l}$, we can replace $\bar m_t$ with $\bar m_t^l$ in \eqref{eqn:gendid}, which in effect enlarges the pool of potential comparison states (when $l=0$, this is often called ``no carryover effects,'' a typical assumption in this literature). One could also combine this with a time-invariance assumption, allowing any state with sequence $\bar m_t^l$ at any $t$ to serve as a comparison state -- or, more weakly, states with the same value of
    $\boldsymbol{\bar{1}}^\top\bar m_t$. Since we use these methods for partial rather than point identification, no such assumption is strictly required; different restrictions on the comparison group simply yield more or less accurate starting estimates of the unobserved counterfactual.
\end{remark}

The challenge in practice is choosing an appropriate causal contrast. For sequences $\bar M_{it} \ne \bar 0_t$, the contrast $Y_{it}(\bar M_{it})-Y_{it}(\bar 0_t)$ is a common choice, giving an ITE that asks: at time $t$, what was the total effect of the observed policy sequence $\bar M_{it}$ relative to no policy change? Alternatively, when $M_{it} = 1$, we may evaluate the sequence $\bar m_t = (\bar M_{i,t-1}, 0)$ -- the effect of treatment in the final period alone, holding the rest of the observed sequence fixed. Interesting choices for untreated (or not-yet-treated) states include $\bar m_t = \bar 1_t$, the effect of having been treated since $T_0$, or $\bar m_t = (\bar 0_{t-1}, 1)$, the effect of turning treatment on at time $t$.

Evaluating any of these contrasts implicitly assumes that we observe some units at the required values of $\bar m_t$; if we do not, we would need to extrapolate from other observed sequences (see, e.g., Remark \ref{rmk:limitdynamic}), which is beyond the scope of this paper. The choice of contrast -- whether a period-specific effect or a long-run cumulative effect -- should be guided by the substantive policy question, and that choice in turn determines the relevant comparison group used to form the starting estimate of the counterfactual. We also note that the no-anticipation assumption maintained throughout is
more demanding in this setting, since not-yet-treated states may anticipate their own future adoption; as elsewhere, the bounding framework can absorb mild violations through the sensitivity parameter.

Finally, once a clear target and working estimate are obtained, we again may choose sensitivity parameters using pre-treatment data. For example, suppose the contrast of interest is $Y_{it}(\bar 1_t)-Y_{it}(\bar 0_t)$ for $t > T_0$ (where one of the two potential outcomes is observed), and let $k = t - T_0$. When $T - T_0$ is small relative to $T_0$ (and hence $k$ is small), one may extend the previously proposed methods by modeling the pre-treatment outcomes $\{Y_{i,k+1},\dots,Y_{i,T_0}\}$ using trends of length $k$ in the comparison group to obtain prediction errors. When $T - T_0$ or $k$ is not small relative to $T_0$, we may instead follow the ``relative magnitude'' bounds of \cite{rambachan2023more}, extrapolating a norm of the consecutive-period changes in pre-treatment prediction errors relative to the final pre-treatment period. A disadvantage of this second approach is that the bounds widen in each subsequent period, which need not occur under the first. Other approaches are of course possible and depend on both the estimand and the method used to form the initial working estimate. Fully exploring these options is beyond the scope of the present paper; we intend this discussion to point generally to how any estimation method, paired with a corresponding sensitivity analysis and pre-treatment prediction errors, can yield bounds as a function of reasonable
sensitivity parameters.

\begin{remark}[Staggered adoption]
    Staggered policy adoption -- where different groups of states adopt policies different times -- is a particularly common instance of this setting. The framework above covers this case, since staggered adoption simply restricts the observed treatment sequences $\bar M_{it}$; in particular, it requires treatment to be irreversible, so that $M_{it} = 1 \implies M_{i,t+k} = 1$ for any $k \ge 1$. This in general rules out some contrasts we might otherwise consider (such as turning treatment off) and limits which states can serve as comparisons at each $t$, but it introduces no conceptual challenge beyond the choice-of-contrast and sensitivity-parameter issues already discussed.
\end{remark}

\subsection{On the width of the bounds}\label{ssec:whenbounds}

While this bounding method is valid for any choice of sensitivity parameters $\hat\tau_i^m \geq \tau_i^{m,\star}$, in some instances even the narrowest defensible bound may contain zero, preventing a researcher from signing an effect. This raises a practical question: what makes a state's bounds more likely to be sign-conclusive (strictly positive or strictly negative)? Assuming we use pre-period data to inform $\hat\tau_i^m$, it is useful to separate features that make sign conclusiveness more likely for a given value of $Z_m$ from features that justify choosing a smaller sensitivity parameter $Z_m$.

\noindent\textbf{Features make sign conclusiveness more likely for a given $Z_m$:}
\begin{itemize}
    \item \emph{Larger-magnitude estimated effects.} A larger point estimate is more
    likely to remain sign-conclusive, since the bound must be wider to reach zero.
    \item \emph{Smaller pre-treatment prediction error.} The better the model predicts the pre-treatment outcomes, the smaller the residuals $\boldsymbol{\tilde\epsilon}_i$ and hence the narrower the resulting bound.\footnote{While we propose a difference-in-differences design, any prediction model is admissible. We do assume the model is chosen out of sample, however, as one can always achieve zero in-sample prediction error.}
    \item \emph{Lower-variance noise.} Lower variance noise will generally result in smaller residuals $\boldsymbol{\tilde\epsilon}_i$, and hence more narrow bounds.
\end{itemize}

\noindent\textbf{Features that justify a smaller $Z_m$:}
\begin{itemize}
    \item \emph{More pre-treatment periods.} Additional pre-treatment data allow the
    typical magnitude of model error to be characterized more reliably, so that a smaller $Z_m$ can be tolerated while still plausibly satisfying $\hat\tau_i^m \ge \tau_i^{m,\star}$.
    \item \emph{Similar comparison states.} When comparison states are otherwise
    comparable in baseline characteristics, the starting counterfactual estimate is more plausible, making a smaller $Z_m$ more justifiable.
\end{itemize}
In practice many states may fail yield sign-conclusive bounds; in our application below (Section \ref{sec:application}), this holds for only two of fifty states. While researchers may regard bounds that contain zero as uninformative, such inconclusive bounds transparently communicate genuine uncertainty about a state-specific effect given the available data and plausible assumptions. We argue that this is preferable to reporting a point estimate with a dubiously estimated confidence interval for a less interpretable causal effect.

\subsection{Simulation}\label{ssec:simulation}

To compare our proposed bounding approach to conventional CATE-based analyses, we use the simulation design described in Appendix \ref{appsec:sim}, essentially the same data-generating process used to construct Figure \ref{fig:illustration}. We generate 1,000 datasets with $N \in \{15,25,50\}$ states observed over $T = 10$ time periods, where treatment occurs in the final period and the realized treatment is $M \in \{0,1,2\}$ (with coarsened indicator $A=\I(M>0)$). Throughout, we restrict attention to datasets with at least three states in each treatment arm. Our goal in this simulation is to compare methods to assess how each method's resulting uncertainty regions compare when trying to learning about the ITE under a repeated sampling framework.

To better align our simulation with real-world analyses, we consider models that do not include $U$ and that estimate effects using the coarsened treatment indicator $A$. To estimate the CATEs, for each dataset, we use parameters estimated from equation (\ref{eqn:linapprox}), using data from the final pre-treatment period and the post-treatment. We emphasize that the targeted quantity of interest in this model is again a linear projection of a weighted combination of effects among the CATEs in each stratum $m = 1, 2$, similar to the illustration in Figure \ref{fig:illustration}.\footnote{For this simulation we reduce the treatment version specific bias of the CATE with respect to the CDE by setting the correlation between $X$ and $U$ to be 0.125.} This is a challenging quantity to interpret generally, but even more in our admittedly adversarial setting, where the slopes of these $m$-specific CATE and CDE functions have different signs. To estimate the ITEs, we use the basic difference-in-differences framework outlined in equation (\ref{eqn:6}), and using the mean absolute error of the pre-treatment residuals to set $\hat\tau_i^{1-A_i}$ as described above, setting $Z_1=Z_0=2$ throughout.  

For all units in each dataset, we generate confidence intervals for the CATE using robust standard error estimates and bounds on the ITE. We then evaluate (1) the coverage for each estimation method with respect to the ITE, and (2) the ability of each method to both correctly exclude zero from the confidence interval or bound, and yield a correctly signed uncertainty region (i.e., does the true ITE have the same sign as the bounds or confidence interval). 

We reiterate that our bounding framework treats all quantities as fixed, with uncertainty arising solely from the fact that the counterfactuals are unobserved. By contrast, this simulation uses a repeated sampling framework. This is purely for evaluation in order to assess how often our bounds would capture the true effects across hypothetical replications with different treatment assignments. These represent conceptually distinct sources of uncertainty: bounds reflect counterfactual uncertainty in a fixed world, while simulations characterize finite-sample performance in a hypothetical super-population. We emphasize, however, that our method provides no formal statistical guaranteesin this repeated sampling framework. Moreover, confidence intervals on the CATE are not designed to provide coverage of the ITE in a repeated sampling framework. We nevertheless report these results because they speak to readers who do operate within a repeated-sampling paradigm: they show that, even judged by frequentist criteria the method does not target, the bounds compare favorably to conventional CATE-based intervals on common metrics such as coverage and sign recovery.

Figure \ref{fig:1} summarizes results across sample sizes ($N = 15,25,50$). Green bars denote ITE bounds and blue bars denote the CATE-based approach; shading distinguishes treated versus untreated units. Relative to the CATE confidence intervals, the ITE bounds generally achieve higher coverage of the true ITE and more reliably recover the ITE sign, though coverage for untreated units is slightly worse at $N = 15,25$. Across methods, performance is consistently better for treated than untreated units, reflecting that (i) predicting counterfactual treated outcomes for controls is harder than predicting counterfactual untreated outcomes for treated units, and (ii) in the simulation design, parallel trends holds on average for treated units but not for untreated units.

\begin{figure}
 \begin{center}
     \includegraphics[scale=0.4]{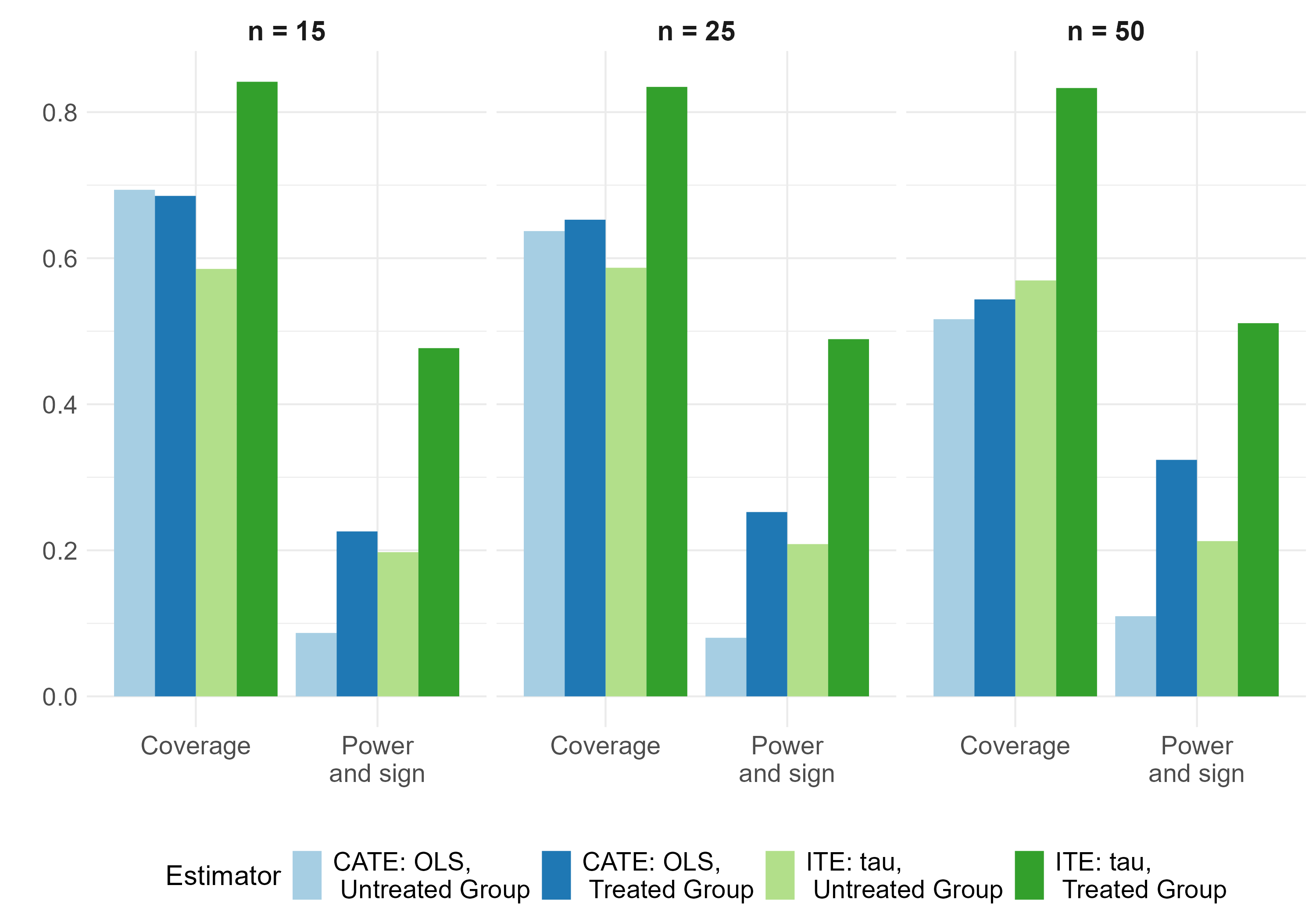}
\end{center}
\caption{Simulation results: coverage of the true ITE (top) and the proportion of correctly signed, zero-excluding uncertainty regions (bottom), comparing data-informed ITE bounds (green) to CATE-based confidence intervals (blue), by treatment status and sample size $N \in \{15, 25, 50\}$.}\label{fig:1}
\end{figure}

Figure \ref{fig:1b} repeats this comparison but uses the infeasible ``oracle'' sensitivity parameter $\tau^{1-A_i,\star}_i$, the smallest value that guarantees coverage. As expected, the resulting bounds have perfect coverage and are substantially more informative, with higher power and correct sign, than the CATE intervals and the data-informed bounds. However, the performance gap remains between treated and control units. This underscores that the method’s performance depends on how well $\tau_i^{1-A_i}$ is chosen. We again emphasize, however, that the CATE intervals target a projected CATE rather than the ITE, so it is not surprising this method does not target the ITE well. Finally, additional results in Appendix \ref{app:results} show that, in this simulation, adding linear covariate adjustment via $\hat\E^{lin}_{1-A_i}$ to the ITE estimators does not improve performance relative to the unconditional approach.

\begin{figure}
 \begin{center}
     \includegraphics[scale=0.4]{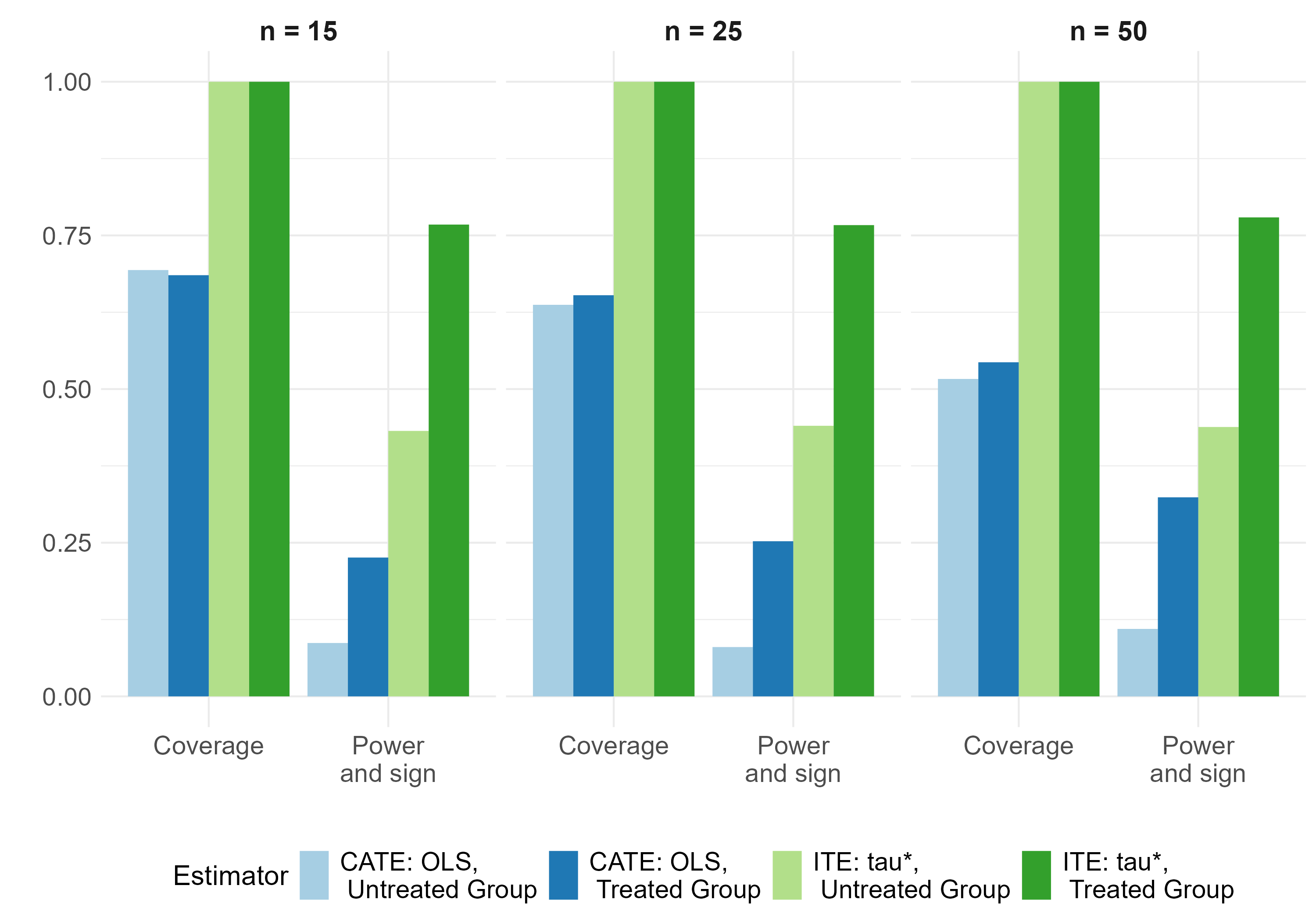}
\end{center}
\caption{Simulation results under the infeasible ``oracle'' sensitivity parameter $\tau_i^{1-A_i,\star}$: coverage (top) and the proportion of correctly signed, zero-excluding uncertainty regions (bottom), by treatment status and sample size. Bounds achieve perfect coverage by construction and are more informative than both the CATE intervals and the data-informed bounds of Figure \ref{fig:1}.}\label{fig:1b}
\end{figure}

\section{Application}\label{sec:application}

We now apply this method to study the effect of Medicaid expansion on high‑volume buprenorphine prescribing. Buprenorphine is the primary medication for opioid use disorder (OUD) in the United States and access remains a major public health concern. Numerous policy efforts have sought to improve access to buprenorphine, including increasing the number of clinicians able to prescribe buprenorphine \cite{saloner2021article}. Nevertheless, recent research has shown that treatment availability largely depends on a small number of high-volume clinicians -- roughly 80\% of all dispensed buprenorphine comes from 17\% of providers \citep{stein2023buprenorphine,schuler2024growing}. Because high-volume prescribers account for the majority of treatment availability, policies that affect their prescribing behavior may have outsized implications for OUD treatment access.

Multiple state policies hypothesized to influence buprenorphine treatment include state Medicaid expansion (which increases insurance coverage) and mandatory PDMP laws (which may impose administrative frictions on prescribing). Prior analyses estimated mean effects of these policies, finding that Medicaid expansion—which, starting in 2014, allowed states to receive federal funding to expand Medicaid eligibility to 138\% of the federal poverty line—was not associated with high-volume prescribing \cite{schuler2025high}. By contrast, mandatory PDMP laws, which require prescribers to check the PDMP before prescribing restricted substances such as opioids, were associated with a decrease in high-volume prescribing \cite{schuler2025high}. However, states differ markedly in their OUD burden, behavioral health infrastructure, and rural provider composition, suggesting that these average effects may mask important state-level heterogeneity. In our application, we reexamine the effect of Medicaid expansion on high-volume prescribing in 2014 using our heterogeneity framework, focusing on two state-level moderators: mandatory PDMP laws and the share of buprenorphine prescribers located in rural communities. This setting exemplifies the practical challenges of studying treatment heterogeneity across small strata of heterogeneous jurisdictions, where point estimates are noisy and policy implementation varies.

We use IQVIA data from 2009-2014 to examine this question, using the state-level total numbers of high-volume prescribers per 100,000 people, with state population estimates coming from the Census Population and Housing Estimates (PEP). For our treatment variable and covariates, we classify the 26 states (including the District of Columbia) that accepted federal funding to expand their Medicaid programs in 2014 as having been treated and 25 that did not as being a control state.\footnote{We omit New Hampshire from consideration, which expanded their Medicaid eligibility requirements in August of 2014. We also note that different studies classify expansion and non-expansion states differently depending on their pre-existing Medicaid eligibility requirements (see, e.g., \cite{rubinstein2023balancing}).}  This immediately raises the treatment coarsening problem, since each state may implement their own Medicaid program differently, and since each state had its own pre-existing Medicaid eligibility requirements. In other words, ``Medicaid expansion'' is effectively a distinct policy change in each state.\footnote{Another way to frame this issue is that the ACA expansion policy standardizes part of the treated regime by setting the eligibility threshold to 138\% of the federal poverty line, while the ``no expansion'' regime corresponds to maintaining each state's pre-2014 Medicaid eligibility rules. In this framework, $Y(M(1))$ may be viewed as the outcome under a partially standardized expansion regime, whereas $Y(M(0))$ reflects a state-specific status quo. While this would contrast with our proposed framework where we assume that $M(0)$ is well-defined and always equal to zero, this framing also highlights the value of considering state-specific effects in this setting, as ``Medicaid expansion'' is inherently state-specific and not a uniform policy intervention.} With respect to mandatory PDMPs, two states had them in 2013 and three more states implemented them in 2014. While there are other potentially relevant policies around the same time that varied across states over this time-frame (for example, continuing medical education requirements), we ignore these for illustrative purposes. Table \ref{tab:2} below lists each state, whether or not they expanded Medicaid in 2014, whether or not they had mandatory PDMPs in 2013 and 2014, and whether more than 25\% of the state's providers were located in rural communities (``rurality'').
\begin{table}[ht]
\centering
\caption{\textbf{Medicaid Expansion and Mandatory PDMP Status by State and Rurality}}
\label{tab:2}
\begin{tabular}{
    >{\centering\arraybackslash}p{2.5cm} 
    >{\centering\arraybackslash}p{1.8cm} 
    >{\centering\arraybackslash}p{1.8cm} 
    >{\centering\arraybackslash}p{1.8cm} 
    p{7cm}}
\toprule
\textbf{Medicaid Expansion} & 
\multicolumn{2}{c}{\textbf{Mandatory PDMP}} &
\multicolumn{1}{c}{\textbf{Rural}} &
\textbf{States} \\
\cmidrule(lr){2-3}
 & \textbf{2014} & \textbf{2013} & \\
\midrule
\midrule
1 & 0 & 0 & 1 & AR, HI, KY, ND, VT\\
1 & 0 & 0 & 0 & AZ, CA, CO, CT, DC, DE, IA, IL, MD, MI, MN, NJ, NM, NV, NY, OH, OR, RI, WA\\
1 & 1 & 0 & 0 & MA\\
1 & 1 & 1 & 0 & WV\\
0 & 0 & 0 & 1 & AK, ID, ME, MS, MT, NE, SD, WY\\
\addlinespace
0 & 0 & 0 & 0 & AL, FL, GA, KS, MO, NC, OK, PA, SC, TX, UT, VA, WI\\
0 & 1 & 0 & 0 & IN, LA\\
0 & 1 & 1 & 0 & TN\\
\bottomrule
\end{tabular}
\end{table}

Given these data constraints, conventional heterogeneity analyses will be limited by these small cell sizes. While the overall ATE can use 50 states (25 treated and 25 control states---a best-case scenario), stratification quickly reduces the available sample: there are 45 states (24 treated) with no PDMPs but only 5 states with any PDMP (2 treated); by rurality, there are 37 non-rural states (21 treated) and 13 rural states (5 treated). Stratification by both PDMP history and rurality is therefore not feasible without borrowing information across groups or imposing additional modeling assumptions.

In Appendix \ref{app:results} we estimate various conventional effects using two-way fixed effects models and quantify uncertainty using cluster robust standard errors and a normal approximation to the sampling distribution. Despite the fact that cluster robust standard errors are known to be biased downwards given small numbers of treated clusters \cite{mackinnon2018wild}, and that the normal approximation does not adjust for the small sample sizes, we find no statistically significant results. This may reflect that there are, in fact, no effects, or the extremely limited power of this analysis.

\subsection{State-specific effects}

We now instead take the approach that our goal is not to estimate average effects, CATEs, or CDEs, but to bound the state-specific effects. We use the difference-in-differences style approach outlined above, constructing the state-specific effect counterfactual estimates using the pre-treatment outcomes for that state plus the term given by the trends in the states of the opposite treatment status, where we include all states in the opposite treatment status in the pool. 

\begin{figure}
    \centering
        \includegraphics[width=\textwidth]{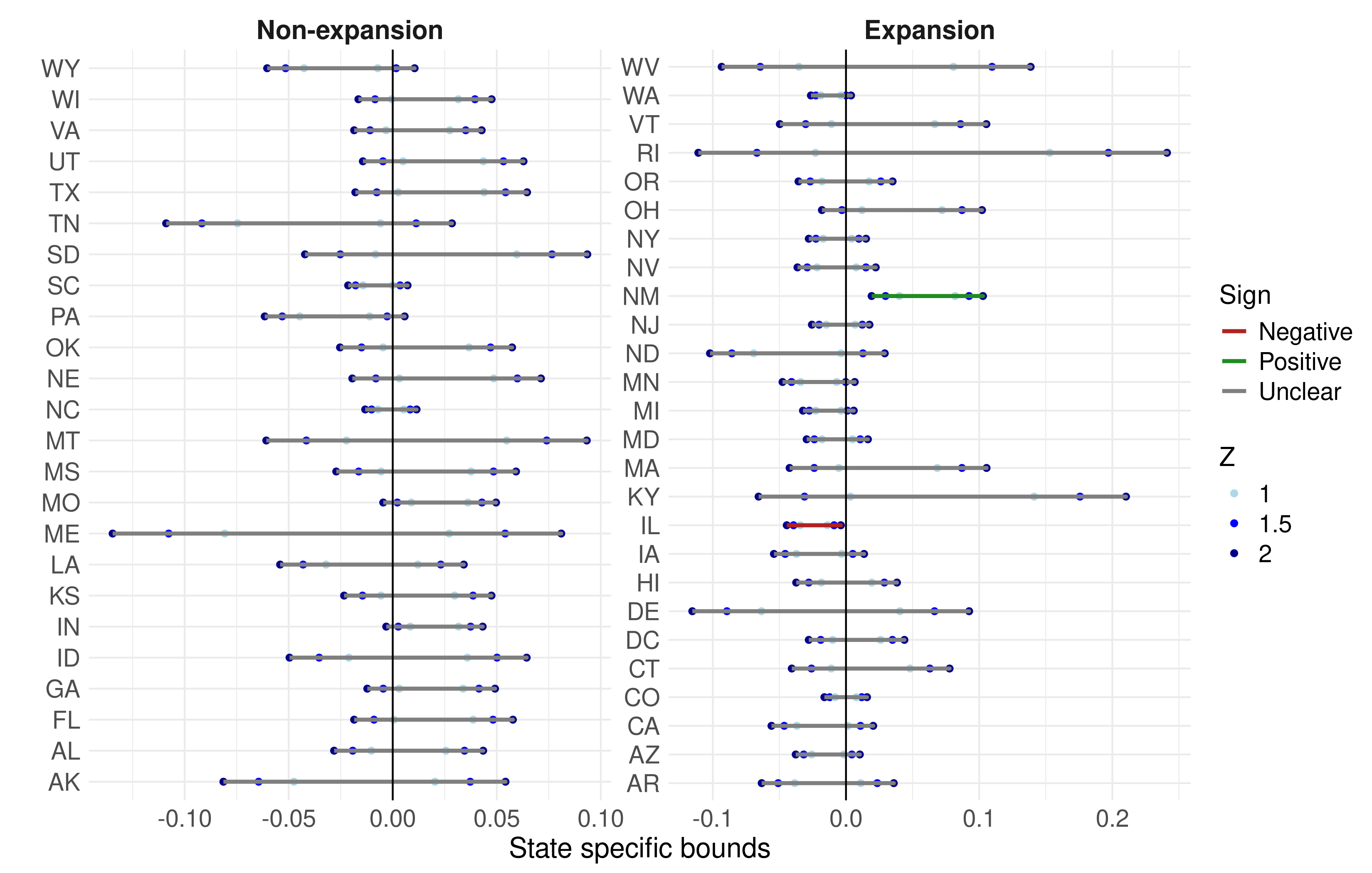}
    \caption{State-specific effect bounds for the effect of 2014 Medicaid expansion on high-volume buprenorphine prescribers per 100,000 people. Dots correspond to sensitivity parameters $Z = 1, 1.5, 2$; bars are colored where the $Z=2$ bounds exclude zero (red: strictly negative; green: strictly positive). Only Illinois (negative) and New Mexico (positive) are sign-conclusive.}
    \label{fig:bounds}
\end{figure}

Figure \ref{fig:bounds} displays state-specific bounds, with dots corresponding to $Z = 1,1.5,2$. Bars are colored by whether the $Z=2$ bounds exclude zero (red for strictly negative effects; green for strictly positive effects). Given the short pre-treatment period, we set $\hat\tau_i^{1-A_i}=Z\|\boldsymbol{\tilde\epsilon}_i\|_\infty$.

Only Illinois shows a strictly negative effect and only New Mexico a strictly positive effect. To assess the robustness of this result, we consider seven alternative specifications: (i) an alternative choice of $\hat\tau_i^{1-A_i}$ based on the final pre-period error plus $Z$ times the maximum change in pre-period errors; (ii) two-way fixed-effects analogues of the primary and alternative first-differences approaches; and (iii) repeating all four approaches after restricting comparison states to those with the same mandatory-PDMP history in 2013 and 2014, yielding eight total specifications. Table \ref{tab:3} counts how often each state is classified as strictly positive or negative: Illinois and New Mexico appear six of eight times, whereas South Carolina (negative) and Missouri (positive) appear three times. All remaining states appear two or fewer times.

\begin{table}
\centering
\caption{Number of strictly positive or negative effects across specifications}
\begin{tabular}{lrr}\toprule\label{tab:3}
State & Negative & Positive\\
\midrule
IL & 6 & 0\\
SC & 3 & 0\\
NM & 0 & 6\\
MO & 0 & 3
\\ \bottomrule\end{tabular}
\end{table}

The robustness of these results strongly suggests that Medicaid expansion had a negative state-specific effect on the number of high-volume prescribers in Illinois and a positive effect in New Mexico. As we have emphasized throughout, effect heterogeneity may broadly be a function of different contextual, compositional, or policy differences. Specifically for this application, the positive state-specific effect observed for New Mexico may reflect a combination of contextual factors, including the state's pre-existing behavioral health infrastructure and prescriber composition. For example, it is possible that Medicaid expansion in New Mexico interacted with a health system already oriented toward financing and delivering OUD treatment, including relatively permissive scope-of-practice rules, such that coverage expansions translated more readily into changes in high-volume prescribing. As a relatively small-population state, modest shifts in provider behavior or patient coverage may translate into comparatively large changes in aggregate outcomes. However, our results, and ITEs generally, cannot isolate underlying mechanisms of heterogeneity. Moreover, while the direction of the bounds appears robust, the magnitude of the true effect is uncertain. Qualitative studies of Medicaid expansion implementation and related opioid policy responses in these states may help elucidate potential sources of these effect patterns. More broadly, the divergent bounds for New Mexico and Illinois illustrate how this method can detect effect heterogeneity even when underlying state-specific effects are only partially identifiable.

\section{Policy implications \& conclusions}

We argue that in many state-policy settings, the CATE may be a poor target of inference, as the CATE cannot generally reveal the true causes of heterogeneity, and therefore has primarily associational interpretations with unclear policy relevance. For example, the CATE cannot directly inform what did or would happen in a particular state under a given policy, nor does it imply that the contexts or policy environments where a policy was observed to be most successful (an association) are the contexts or policy environments that would enhance a policy's efficacy (a causal claim). Unfortunately, the more policy-relevant state-specific and controlled direct effects are not generally point identified without stronger and often implausible assumptions. Moreover, in state policy settings, we often simply lack the data to estimate the CATE or CDE well: the estimators are often likely misspecified or underpowered, and conventional methods for uncertainty quantification are unreliable regardless. 

Motivated by these limitations, we instead propose bounding state-specific treatment effects as an alternative goal and outline one approach to do so under a difference-in-differences framework. Importantly, other approaches are possible under other causal identification strategies and perhaps are more desirable in some settings. Regardless, this general method targets a well-defined and often policy-relevant causal estimand, the state-specific effect. This effect may be interpreted as: what would or did happen to a specific state under a specific intervention. The bounds give ranges for each state’s effect under explicit and interpretable assumptions, clarifying which effect estimates are robust to plausible deviations from idealized identifying conditions and which are not. We also consider using this method in the presence of treatment coarsening, a common scenario for state policy analyses that complicates the interpretation of CATEs. We consider a causal framework similar to \cite{vanderweele2013causal} to define state-specific effects for untreated states even in the presence of treatment coarsening, and explore the use of pre-treatment data, when available, to inform these bounds, extending methods proposed by \cite{manski2018right, rambachan2023more}. Through simulations, we show that when the estimand of interest is the state-specific effect, our proposed bounds can outperform conventional heterogeneity analyses in terms of coverage and the ability to correctly identify the sign of an effect. While state-specific effects cannot directly tell researchers about what variables or policies drive heterogeneity, they nevertheless provide policy-relevant information in settings where point identification or reliable statistical inference for the CATE or CDE is not possible. Finally, we illustrate this method by examining the effect of the Affordable Care Act Medicaid expansion on high-volume buprenorphine prescribing. The application highlights how bounding can yield informative conclusions about state-specific effects even when conventional heterogeneity analyses lack power or yield results that are difficult to interpret.

Other methods exist to estimate state-specific effects, including synthetic controls approaches \citep{abadie2010synthetic, abadie2015comparative} and Bayesian hierarchical models \cite{mcglothlin2018bayesian}. Bayesian models can generate state-level estimates when sample sizes are small \cite{gelman2007data}, but rely on strong structural and distributional assumptions, including the specified prior and likelihood function, to determine the distribution of these effect estimates. The resulting Bayesian credible intervals can then quantify uncertainty conditional on the model, but not uncertainty about whether the model is correctly specified. Our bounding approach instead characterizes the range of state-specific effects consistent with the observed data under weak assumptions, providing an alternative approach to bounding ITEs when researchers wish to minimize reliance on unverifiable modeling choices. Synthetic control methods and their extensions have yielded estimators that are consistent for the ITE \cite{abadie2010synthetic} for both treated and untreated units \cite{agarwal2026synthetic}, with \cite{chernozhukov2021exact} even showing how to construct valid prediction sets for the ITE without assuming the model is correctly specified. Relatedly, comparative interrupted time-series (CITS) designs \citep{fry2021birds} augment the treated series with a comparison group to absorb shared temporal shocks, yielding a counterfactual trajectory for a treated unit; this logic is closely related to the difference-in-differences framework underlying our own working estimate -- though whereas CITS approaches attempt to model the pre-treatment deviations directly, our approach instead uses them only to bound the counterfactual error. Both synthetic control and CITS designs are generally most attractive when a longer pre-treatment series is available, and both typically assume a well-defined intervention. By contrast, our proposed approach works even with only one pre-treatment time period (though we cannot calibrate the sensitivity parameter), and could easily be extended to the setting with no pre-treatment data by making working estimates for the unobserved counterfactuals using cross-sectional approaches. Importantly, well-defined interventions are not the norm in state-policy settings, which, as we have shown, necessitates special care when defining ITEs with respect to non-treated units.

This paper has some broad limitations that present opportunities for future research. First, we only consider state-level policy analyses with state-level data. It is possible that with more granular data, including at the county or even individual level, more meaningful heterogeneity analyses may be conducted, though statistical inference will still remain challenging when the exposure is at the state-level, given that this is arguably the most natural source of uncertainty. Second, we closely follow the proposed methods of \cite{manski2018right, rambachan2023more} for our partial identification approach, and so our proposed methods share the same limitation that these methods may be somewhat conservative. For example, in settings with a large number of pre-treatment periods, one may also consider autoregressive approaches (or, relatedly, CITS designs) to model the time-series evolution of pre-treatment parallel-trend violations as a way to better inform the bounds \citep{kwon2024empirical, han2025bayesian, fry2021birds}. Finally, we have only considered retrospective effect estimation. Arguably, the most relevant quantity is some contemporaneous or future effect (see, e.g., \cite{deb2025counterfactual}). Additional research is needed to better formalize and develop methods for this more general setting.

In conclusion, applied researchers should not abandon heterogeneity analyses in state-policy settings, but should choose their inferential goals carefully. They should be explicit about which causal estimands are of interest, what assumptions are required to identify them, and whether the available data are sufficient to support credible identification and estimation of these effects. In many state-policy settings, we argue that partial identification of state-specific effects offers a principled and transparent alternative to conventional heterogeneity analyses and may better align causal inference with the questions both analysts and policymakers seek to answer.

\section{Acknowledgments}

The authors thank the two anonymous reviewers for their suggestions which greatly improved the quality of this manuscript.

The authors also acknowledge using ChatGPT (OpenAI) and Claude Code for assistance with editing, organization, and checking mathematical arguments and proofs, as well as Cursor to document and test the code used for the simulations and analyses. The authors are solely responsible for the analysis and conclusions in this paper.

\textbf{Funding information:} This research was financially supported through a National Institutes of Health (NIH) grant (P50DA046351) to RAND (PI: Bradley D. Stein). NIH had no role in the design of the study, analysis, and interpretation of data nor in writing the manuscript.

\textbf{Conflict of interest:} The authors state no conflicts of interest.

\textbf{Data availability statement:} The IQVIA data used for the application is not publicly available. However, all code used for both the simulations and application is available, along with a partially simulated dataset that replaces the IQVIA data with simulated data, at \newline \texttt{https://github.com/mrubinst757/state-heterogeneity}.

\bibliographystyle{plain}
\bibliography{bibbb.bib}

@article{stein2023buprenorphine,
  title={Buprenorphine treatment episodes during the first year of COVID: a retrospective examination of treatment initiation and retention},
  author={Stein, Bradley D and Landis, Rachel K and Sheng, Flora and Saloner, Brendan and Gordon, Adam J and Sorbero, Mark and Dick, Andrew W},
  journal={Journal of general internal medicine},
  volume={38},
  number={3},
  pages={733--737},
  year={2023},
  publisher={Springer}
}

@article{mcglothlin2018bayesian,
  title={Bayesian hierarchical models},
  author={McGlothlin, Anna E and Viele, Kert},
  journal={Jama},
  volume={320},
  number={22},
  pages={2365--2366},
  year={2018},
  publisher={American Medical Association}
}

@article{vanderweele2013causal,
  title={Causal inference under multiple versions of treatment},
  author={VanderWeele, Tyler J and Hernan, Miguel A},
  journal={Journal of causal inference},
  volume={1},
  number={1},
  pages={1--20},
  year={2013},
  publisher={De Gruyter}
}

@techreport{garthwaite2019all,
  title={All Medicaid expansions are not created equal: the geography and targeting of the Affordable Care Act},
  author={Garthwaite, Craig and Graves, John A and Gross, Tal and Karaca, Zeynal and Marone, Victoria R and Notowidigdo, Matthew J},
  year={2019},
  institution={National Bureau of Economic Research}
}

@article{mcginty2024scaling,
  title={Scaling interventions to manage chronic disease: innovative methods at the intersection of health policy research and implementation science},
  author={McGinty, Emma E and Seewald, Nicholas J and Bandara, Sachini and Cerd{\'a}, Magdalena and Daumit, Gail L and Eisenberg, Matthew D and Griffin, Beth Ann and Igusa, Tak and Jackson, John W and Kennedy-Hendricks, Alene and others},
  journal={Prevention Science},
  volume={25},
  number={Suppl 1},
  pages={96--108},
  year={2024},
  publisher={Springer}
}

@article{schuler2021methodological,
  title={Methodological challenges and proposed solutions for evaluating opioid policy effectiveness},
  author={Schuler, Megan S. and Griffin, Beth Ann and Cerd{\'a}, Magdalena and McGinty, Emma E. and Stuart, Elizabeth A.},
  journal={Health Services and Outcomes Research Methodology},
  volume={21},
  pages={21--41},
  year={2021},
  publisher={Springer}
}

@article{schuler2020state,
  title={The state of the science in opioid policy research},
  author={Schuler, Megan S. and Heins, Sara E. and Smart, Rosanna and Griffin, Beth Ann and Powell, David and Stuart, Elizabeth A. and Pardo, Bryce and Smucker, Sierra and Patrick, Stephen W. and Pacula, Rosalie Liccardo and Stein, Bradley D.},
  journal={Drug and Alcohol Dependence},
  volume={214},
  pages={108137},
  year={2020},
  publisher={Elsevier}
}

@article{hettinger2025causal,
  title={A causal framework for evaluating drivers of policy effect heterogeneity using difference-in-differences},
  author={Hettinger, Gary and Lee, Youjin and Mitra, Nandita},
  journal={Health Services and Outcomes Research Methodology},
  pages={1--22},
  year={2025},
  publisher={Springer}
}

@article{schuler2025high,
  title={High-volume buprenorphine prescribers: examining state policy contexts},
  author={Schuler, Megan S and Sheng, Flora and Saloner, Brendan and Gordon, Adam J and Stein, Bradley D},
  journal={Drug and Alcohol Dependence Reports},
  pages={100406},
  year={2026},
  publisher={Elsevier}
}

@book{gelman2007data,
  title={Data analysis using regression and multilevel/hierarchical models},
  author={Gelman, Andrew and Hill, Jennifer},
  year={2007},
  publisher={Cambridge university press}
}

@article{hill2011bayesian,
  author    = {Jennifer L. Hill},
  title     = {Bayesian Nonparametric Modeling for Causal Inference},
  journal   = {Journal of Computational and Graphical Statistics},
  volume    = {20},
  number    = {1},
  pages     = {217--240},
  year      = {2011},
  doi       = {10.1198/jcgs.2010.08162},
  publisher = {Taylor & Francis}
}

@incollection{pearl2022direct,
  title={Direct and indirect effects},
  author={Pearl, Judea},
  booktitle={Probabilistic and causal inference: the works of Judea Pearl},
  pages={373--392},
  year={2022}
}

@techreport{callaway2024difference,
  title={Difference-in-differences with a continuous treatment},
  author={Callaway, Brantly and Goodman-Bacon, Andrew and Sant'Anna, Pedro HC},
  year={2024},
  institution={National Bureau of Economic Research}
}

@article{antonelli2024autoregressive,
  title={Autoregressive models for panel data causal inference with application to state-level opioid policies},
  author={Antonelli, Joseph and Rubinstein, Max and Agniel, Denis and Smart, Rosanna and Stuart, Elizabeth and Cefalu, Matthew and Schell, Terry and Eagan, Joshua and Stone, Elizabeth and Griswold, Max and others},
  journal={arXiv preprint arXiv:2408.09012},
  year={2024}
}

@article{smart2024investigating,
  title={Investigating the complexity of naloxone distribution: Which policies matter for pharmacies and potential recipients},
  author={Smart, Rosanna and Powell, David and Pacula, Rosalie Liccardo and Peet, Evan and Abouk, Rahi and Davis, Corey S},
  journal={Journal of health economics},
  volume={97},
  pages={102917},
  year={2024},
  publisher={Elsevier}
}

@article{rubinstein2023heterogeneous,
  title={Heterogeneous interventional effects with multiple mediators: Semiparametric and nonparametric approaches},
  author={Rubinstein, Max and Branson, Zach and Kennedy, Edward H},
  journal={Journal of Causal Inference},
  volume={11},
  number={1},
  pages={20220070},
  year={2023},
  publisher={De Gruyter}
}

@book{angrist2009mostly,
  title={Mostly harmless econometrics: An empiricist's companion},
  author={Angrist, Joshua D and Pischke, J{\"o}rn-Steffen},
  year={2009},
  publisher={Princeton university press}
}

@article{buja2019models,
  title={Models as approximations I},
  author={Buja, Andreas and Brown, Lawrence and Berk, Richard and George, Edward and Pitkin, Emil and Traskin, Mikhail and Zhang, Kai and Zhao, Linda},
  journal={Statistical Science},
  volume={34},
  number={4},
  pages={523--544},
  year={2019},
  publisher={JSTOR}
}

@article{heiler2024heterogeneous,
  author    = {Phillip Heiler},
  title     = {Heterogeneous Treatment Effect Bounds under Sample Selection with an Application to the Effects of Social Media on Political Polarization},
  journal   = {Journal of Econometrics},
  volume    = {244},
  number    = {1},
  pages     = {105856},
  year      = {2024},
  doi       = {10.1016/j.jeconom.2023.105856},
  publisher = {Elsevier}
}

@article{holland1986statistics,
  title={Statistics and causal inference},
  author={Holland, Paul W},
  journal={Journal of the American statistical Association},
  volume={81},
  number={396},
  pages={945--960},
  year={1986},
  publisher={Taylor \& Francis}
}

@article{rubinstein2023balancing,
  title={Balancing weights for region-level analysis: The effect of Medicaid expansion on the uninsurance rate among states that did not expand Medicaid},
  author={Rubinstein, Max and Haviland, Amelia and Choi, David},
  journal={The Annals of Applied Statistics},
  volume={17},
  number={2},
  pages={1469--1490},
  year={2023},
  publisher={Institute of Mathematical Statistics}
}

@article{agarwal2026synthetic,
  title={Synthetic interventions: Extending synthetic controls to multiple treatments},
  author={Agarwal, Anish and Shah, Devavrat and Shen, Dennis},
  journal={Operations Research},
  volume={74},
  number={2},
  pages={840--859},
  year={2026},
  publisher={INFORMS}
}

@article{abadie2010synthetic,
  title={Synthetic control methods for comparative case studies: Estimating the effect of California’s tobacco control program},
  author={Abadie, Alberto and Diamond, Alexis and Hainmueller, Jens},
  journal={Journal of the American statistical Association},
  volume={105},
  number={490},
  pages={493--505},
  year={2010},
  publisher={Taylor \& Francis}
}

@article{deb2025counterfactual,
  title={Counterfactual Forecasting for Panel Data},
  author={Deb, Navonil and Dwivedi, Raaz and Basu, Sumanta},
  journal={arXiv preprint arXiv:2511.06189},
  year={2025}
}

@article{chernozhukov2021exact,
  title={An exact and robust conformal inference method for counterfactual and synthetic controls},
  author={Chernozhukov, Victor and W{\"u}thrich, Kaspar and Zhu, Yinchu},
  journal={Journal of the American Statistical Association},
  volume={116},
  number={536},
  pages={1849--1864},
  year={2021},
  publisher={Taylor \& Francis}
}

@article{abadie2015comparative,
  title={Comparative politics and the synthetic control method},
  author={Abadie, Alberto and Diamond, Alexis and Hainmueller, Jens},
  journal={American Journal of Political Science},
  volume={59},
  number={2},
  pages={495--510},
  year={2015},
  publisher={Wiley Online Library}
}

@article{tas2019should,
  title={Should we worry that take-home naloxone availability may increase opioid use?},
  author={Tas, Basak and Humphreys, Keith and McDonald, Rebecca Silvia and Strang, John S},
  journal={Addiction},
  volume={114},
  number={10},
  pages={1723--1725},
  year={2019},
  publisher={Wiley-Blackwell Publishing Ltd}
}

@article{goodman2021difference,
  title={Difference-in-differences with variation in treatment timing},
  author={Goodman-Bacon, Andrew},
  journal={Journal of econometrics},
  volume={225},
  number={2},
  pages={254--277},
  year={2021},
  publisher={Elsevier}
}

@article{fry2021birds,
  title={Birds of a feather flock together: Comparing controlled pre--post designs},
  author={Fry, Carrie E and Hatfield, Laura A},
  journal={Health Services Research},
  volume={56},
  number={5},
  pages={942--952},
  year={2021},
  publisher={Wiley Online Library}
}

@inproceedings{kwon2024empirical,
  title={(Empirical) Bayes Approaches to Parallel Trends},
  author={Kwon, Soonwoo and Roth, Jonathan},
  booktitle={AEA Papers and Proceedings},
  volume={114},
  pages={606--609},
  year={2024},
  organization={American Economic Association 2014 Broadway, Suite 305, Nashville, TN 37203}
}

@article{han2025bayesian,
  title={Bayesian Sensitivity Analyses for Policy Evaluation with Difference-in-Differences under Violations of Parallel Trends},
  author={Han, Seong Woo and Mitra, Nandita and Hettinger, Gary and Oganisian, Arman},
  journal={arXiv preprint arXiv:2508.02970},
  year={2025}
}

@article{hettinger2025multiply,
  title={Multiply robust difference-in-differences estimation of causal effect curves for continuous exposures},
  author={Hettinger, Gary and Lee, Youjin and Mitra, Nandita},
  journal={Biometrics},
  volume={81},
  number={1},
  pages={ujaf015},
  year={2025},
  publisher={Oxford University Press}
}

@article{zhang2025continuous,
  title={Continuous difference-in-differences with double/debiased machine learning},
  author={Zhang, Lucas Zheng},
  journal={Econometrics Journal},
  pages={utaf024},
  year={2025},
  publisher={Oxford University Press}
}

@article{seewald2024target,
  title={Target trial emulation for evaluating health policy},
  author={Seewald, Nicholas J and McGinty, Emma E and Stuart, Elizabeth A},
  journal={Annals of internal medicine},
  volume={177},
  number={11},
  pages={1530--1538},
  year={2024},
  publisher={American College of Physicians}
}

@article{heiler2022heterogeneity,
  title={Effect or treatment heterogeneity? Policy evaluation with aggregated and disaggregated treatments},
  author={Heiler, Phillip and Knaus, Michael},
  year={2022},
  publisher={JSTOR}
}

@misc{chernozhukov2018double,
  title={Double/debiased machine learning for treatment and structural parameters},
  author={Chernozhukov, Victor and Chetverikov, Denis and Demirer, Mert and Duflo, Esther and Hansen, Christian and Newey, Whitney and Robins, James},
  year={2018},
  publisher={Oxford University Press Oxford, UK}
}

@article{mackinnon2018wild,
  title={The wild bootstrap for few (treated) clusters},
  author={MacKinnon, James G and Webb, Matthew D},
  journal={The Econometrics Journal},
  volume={21},
  number={2},
  pages={114--135},
  year={2018},
  publisher={Oxford University Press Oxford, UK}
}

@article{schuler2024growing,
  title={Growing importance of high-volume buprenorphine prescribers in OUD treatment: 2009--2018},
  author={Schuler, Megan S and Dick, Andrew W and Gordon, Adam J and Saloner, Brendan and Kerber, Rose and Stein, Bradley D},
  journal={Drug and alcohol dependence},
  volume={259},
  pages={111290},
  year={2024},
  publisher={Elsevier}
}

@article{saloner2021article,
  title={Article commentary: It will end in tiers: A strategy to include “Dabblers” in the buprenorphine workforce after the X-waiver},
  author={Saloner, Brendan and Andraka Christou, Barbara and Gordon, Adam J and Stein, Bradley D},
  journal={Substance Abuse},
  volume={42},
  number={2},
  pages={153--157},
  year={2021},
  publisher={SAGE Publications Sage CA: Los Angeles, CA}
}

@article{vegetabile2021distinction,
  title={On the distinction between ``conditional average treatment effects" (CATE) and ``individual treatment effects" (ITE) under ignorability assumptions},
  author={Vegetabile, Brian G},
  journal={arXiv preprint arXiv:2108.04939},
  year={2021}
}

@article{kennedy2022semiparametric,
  title={Semiparametric doubly robust targeted double machine learning: a review},
  author={Kennedy, Edward H},
  journal={arXiv preprint arXiv:2203.06469},
  year={2022}
}

@article{manski2018right,
  title={How do right-to-carry laws affect crime rates? Coping with ambiguity using bounded-variation assumptions},
  author={Manski, Charles F and Pepper, John V},
  journal={Review of Economics and Statistics},
  volume={100},
  number={2},
  pages={232--244},
  year={2018},
  publisher={MIT Press}
}

@article{rambachan2023more,
  title={A more credible approach to parallel trends},
  author={Rambachan, Ashesh and Roth, Jonathan},
  journal={Review of Economic Studies},
  volume={90},
  number={5},
  pages={2555--2591},
  year={2023},
  publisher={Oxford University Press US}
}

@article{mackinnon2017wild,
  title={Wild bootstrap inference for wildly different cluster sizes},
  author={MacKinnon, James G and Webb, Matthew D},
  journal={Journal of Applied Econometrics},
  volume={32},
  number={2},
  pages={233--254},
  year={2017},
  publisher={Wiley Online Library}
}

@article{imbens2004nonparametric,
  title={Nonparametric estimation of average treatment effects under exogeneity: A review},
  author={Imbens, Guido W},
  journal={Review of Economics and statistics},
  volume={86},
  number={1},
  pages={4--29},
  year={2004},
  publisher={MIT Press 238 Main St., Suite 500, Cambridge, MA 02142-1046, USA journals~…}
}

@article{almirall2016adaptive,
  title={Adaptive interventions in child and adolescent mental health},
  author={Almirall, Daniel and Chronis-Tuscano, Andrea},
  journal={Journal of Clinical Child \& Adolescent Psychology},
  volume={45},
  number={4},
  pages={383--395},
  year={2016},
  publisher={Taylor \& Francis}
}

@article{mackinnon2020randomization,
  title={Randomization inference for difference-in-differences with few treated clusters},
  author={MacKinnon, James G and Webb, Matthew D},
  journal={Journal of Econometrics},
  volume={218},
  number={2},
  pages={435--450},
  year={2020},
  publisher={Elsevier}
}

@article{gardner2022two,
  title={Two-stage differences in differences},
  author={Gardner, John},
  journal={arXiv preprint arXiv:2207.05943},
  year={2022}
}

@article{callaway2021difference,
  title={Difference-in-differences with multiple time periods},
  author={Callaway, Brantly and Sant’Anna, Pedro HC},
  journal={Journal of econometrics},
  volume={225},
  number={2},
  pages={200--230},
  year={2021},
  publisher={Elsevier}
}

\newpage

\appendix

\section{Theoretic results}\label{app:proofs}

\subsection{Identification results: treatment coarsening}\label{appsec:estimands}

We consider causal identification in the presence of treatment coarsening from a potential outcomes framework; first using an ignorability, or selection-on-observables, framework; and second, using a difference-in-differences framework. We conclude by considering the causal interpretation of the unit-level DiD estimators proposed in Section \ref{ssec:extensionscoarsening} in the presence of treatment coarsening.

\subsubsection{Ignorability}\label{appsec:coarsening}

We begin by presenting Proposition \ref{prop:rep}, which replicates a result from \citep{hettinger2025causal} that gives a causal interpretation for the estimand $\E[Y \mid A = 1, x] - \E[Y \mid A = 0, x]$, a standard estimator of the CATE \cite{kennedy2022semiparametric}, under assumptions \ref{asmpt:1}-\ref{asmpt:3} below. We then show how, under additional assumptions, we obtain the same expression for the identified causal estimand when defining $Y(A=1)=Y(M(1))$ as we propose, and formalized in Proposition \ref{prop:equiv}. To motivate Proposition \ref{prop:rep}, we make the following assumptions:

\stepcounter{assumption} 

\begin{intassumption}[Consistency]\label{asmpt:1}
    \begin{align*}
        Y &= \sum_{m=0}^k\I(M = m)Y(m). \\
    \end{align*}
\end{intassumption}

\noindent Consistency states that the observed value of $Y$ is equal to the potential outcome under the respective treatment.

\begin{intassumption}[Ignorability]\label{asmpt:2}
    \begin{align*}
    Y(m) &\perp M \mid X, \qquad m = 0,\dots,k.
    \end{align*}
\end{intassumption}

\noindent Ignorability states that $M$ is effectively randomized with respect to $Y(m)$ conditional on the observed covariates $X$.

\begin{intassumption}[Positivity]\label{asmpt:3}
    \begin{align*}
    \Pr(\min_a\Pr(A = a \mid X) &> \epsilon) = 1, \qquad \epsilon > 0. \\
    \end{align*}
\end{intassumption}

\noindent Positivity states that there is some positive probability of receiving any treatment-level defined by the coarsened treatment indicator $A$.

Proposition \ref{prop:rep} states that the standard CATE expression under treatment coarsening gives a weighted average of CATEs $\tilde\psi_m(x)$, with weights proportional to the conditional distribution of each value of $m$ for $m > 0$. This result was previously shown in \cite{hettinger2025causal}.

\begin{proposition}\label{prop:rep}
    Under assumptions \ref{asmpt:1}-\ref{asmpt:3},
    \begin{align*}
        \E[Y \mid A = 1, x] - \E[Y \mid A = 0, x] &= \sum_{m > 0}\E[Y(m)\mid x]\frac{\Pr(M =m \mid x)}{\sum_{m >0}\Pr(M = m \mid x)} - \E[Y(0) \mid A = 0,x]
    \end{align*}    
\end{proposition}

\begin{proof}
    \begin{align*}
        \E[Y \mid A = 1, x] &= \sum_{m>0}\E[Y \mid M = m, A = 1, x]\Pr(M = m \mid A = 1, x),
    \end{align*}

    \noindent where we sum only over $m>0$ since $M>0 \iff A=1$. Then,

    \begin{align*}
        \E[Y \mid A = 1, x] - \E[Y \mid A = 0, x] &= \sum_{m > 0} \E[Y(m) \mid M = m, x]\Pr(M = m \mid A =  1, x) - \E[Y(0) \mid x] \\
        &=\sum_{m > 0} \E[Y(m) \mid x]\frac{\Pr(M =m \mid x)}{\sum_{m >0}\Pr(M = m \mid x)} - \E[Y(0) \mid x] \\
        &= \sum_{m > 0} \E[Y(m)-Y(0) \mid x]\frac{\Pr(M =m \mid x)}{\sum_{m >0}\Pr(M = m \mid x)}.
    \end{align*}

    \noindent where we use outcome consistency to replace $Y$ by $Y(m)$, and ignorability to drop conditioning on $M$ inside the expectation.
\end{proof}

\begin{remark}
    The positivity requirement is somewhat weaker than what would be required to identify $\E[Y(m) \mid x]$ for all $m$, which would require $\Pr(\min_m \Pr(M = m\mid X) > \epsilon) = 1$ for some $\epsilon > 0$.
\end{remark}

We now formalize our proposed framework, which may be viewed as a special case of the one proposed in \cite{vanderweele2013causal}, where we posit that $(A, M)$ is a sequentially randomized trial where $A = 0 \implies M = 0$ and $M$ is randomized to values $m \in \{1,\dots,k\}$ when $A = 1$. 

\stepcounter{assumption}

\begin{intassumption}[Exclusion restriction]\label{asmpt:0}
    \begin{align*}
    Y(a) = Y(a, M(a)) = Y(M(a)), \qquad \forall a = 0, 1.       
    \end{align*}
\end{intassumption}

\noindent The exclusion restriction states that $A$ only affects $Y$ through $M$. This follows by definition since all information about $A$ is contained in $M$.

\begin{intassumption}[Y-M Consistency]\label{asmpt:2a}
    \begin{align*}
        Y &= \sum_{m=0}^k\I(M = m)Y(m), \\
        M &= AM(1), \qquad M(1) = 1,\dots,k.
    \end{align*}
\end{intassumption}

\noindent Y-M consistency adds a consistency statement with respect to $M$ in addition to $Y$, stating that the observed value of $M$ is equal to $M(1)$ if $A = 1$ (a quantity greater than 0), and equal to zero otherwise. This again implies that $A = 0 \iff M = 0$, and that $M(0) = 0$ and is a well-defined single policy state. The consistency assumption with respect to $Y$ remains unchanged from before.

\begin{intassumption}[Ignorability]\label{asmpt:3a}
    \begin{align*}
    Y(M(a)) &\perp A \mid X \qquad \forall a = 0, 1.
    \end{align*}
\end{intassumption}

\noindent Ignorability states that $A$ is effectively randomized with respect to $Y(a)$ given $X$ for $a = 0, 1$.

\begin{intassumption}[Positivity]\label{asmpt:4a}
    \begin{align*}
        \Pr(\min_a\Pr(A = a \mid X) > \epsilon) = 1, \qquad \epsilon > 0.
    \end{align*}
\end{intassumption}

\noindent Finally, we assume positivity, which is equivalent to what we had defined previously. Proposition \ref{prop:equiv} shows that this framework gives an identical identified expression for the CATE.

\begin{proposition}\label{prop:equiv}
    Under assumptions \ref{asmpt:0}-\ref{asmpt:4a},
    \begin{align*}
    \E[Y(A = 1)-Y(A = 0) \mid x] &=  \E[Y \mid x, A = 1]- \E[Y \mid x, A = 0].
    \end{align*}
\end{proposition}

\begin{proof}
\begin{align*}
    \E[Y(A=1)-Y(A=0) \mid x] &= \E[Y(M(1)) - Y(M(0)) \mid x] \\
    &= \E[Y(M(1)) \mid x, A = 1] - \E[Y(M(0)) \mid x, A = 0] \\
    &= \E[Y(M) \mid x, A = 1] - \E[Y(0) \mid x, A = 0] \\
    &= \E[Y \mid x, A = 1] - \E[Y \mid x, A = 0]
\end{align*}

\noindent where the first equality holds by definition and the exclusion restriction, the second by ignorability, and the third and fourth lines by consistency and the fact that $M = 0 \iff A = 0$. \end{proof}

\begin{remark}
    Proposition \ref{prop:equiv} shows that our proposed definition of $Y(A = 1)$ yields an identical identified expression for the CATE that is commonly estimated in practice using a randomization-based (or ``selection-on-observables'') causal identification framework. While Proposition \ref{prop:equiv} requires additional assumptions to define $Y(A = 1)$ relative to Proposition \ref{prop:rep}, identification of this expression is possible even when $Y(m) \not\perp M \mid X, A = 1$, which is precluded by assumption \ref{asmpt:2}. 
    
    However, this expression is \textit{not} equal to a weighted average of the CATEs $\tilde\psi_m(x)$ unless we invoke an additional ignorability assumption. To see this notice that,
    
    \begin{align*}
        \E[Y \mid x, A = 1] &= \sum_{m>0}\E[Y(m) \mid M = m, A = 1, x]\Pr(M = m \mid A = 1,x) \\
        &\ne \sum_{m>0}\E[Y(m) \mid x]\Pr(M = m \mid A = 1, x). 
    \end{align*}

    The second line would hold, for example, if we invoked assumption \ref{asmpt:2}. Finally, if assumption \ref{asmpt:2} does hold, notice that this last equality suggests that we may alternatively view $\E[Y \mid x, A = 1] - \E[Y \mid x, A = 0]$ as representing a stochastic intervention: $\E[Y(G_1(m)) - Y(0) \mid x]$ where $G_1(m)$ is a random draw from $\P(M \mid A = 1, x)$. This may be a preferable interpretation of this quantity in settings where it is not acceptable to view $A$ as preceding $M$. We refer to \cite{vanderweele2013causal} for more discussion of stochastic interventions in this setting.
\end{remark}

\begin{remark}
    Proposition \ref{prop:equiv} is essentially a special case of Proposition 3 from \cite{vanderweele2013causal}, where we consider a binary $A$ and do not allow for additional confounders of the $(A, M)$ relationship.
\end{remark}

\begin{remark}
    We have only considered identification of the CATE; however, extending this result to the CDE is straightforward. We would simply define the potential outcomes $Y(m,x)$, and assume the relevant ignorability (possibly conditional on other covariates $Z$) and consistency assumptions to obtain the equality $\E[Y(m) \mid X = x] = \E[Y(m,x)]$ (or $\E[Y(m) \mid X = x, Z = z] = \E[Y(m,x) \mid Z = z]$). While these assumptions are stronger than the ones required to identify the CATE, the identified estimand will be identical (assuming $Z = \emptyset$; the expression will simply condition on $Z$ otherwise).
\end{remark}

\subsubsection{Difference-in-differences}\label{appsec:didcoarsening}

We now consider causal identification under treatment coarsening in a difference-in-differences framework. For simplicity, we let $t \in \{0, 1\}$ and assume that treatment only occurs at time $t = 1$. We define the conditional difference-in-differences estimand,

\begin{align*}
   \psi^{\mathrm{DiD}}(x) &= \E[Y_{i1}-Y_{i0} \mid A = 1, X = x] - \E[Y_{i1}-Y_{i0} \mid A = 0, X = x].
\end{align*}

\noindent We show the causal estimands this quantity identifies under different assumptions. We condition on $X = x$ throughout, but briefly observe that we may take $X = \emptyset$ throughout to obtain unconditional analogues of all assumptions and quantities, or take $\E[\psi^{\mathrm{DiD}}(X)]$ to obtain results that hold averaged over the covariate distribution. We will continue to invoke assumptions \ref{asmpt:0} (the exclusion restriction) and \ref{asmpt:2a} (consistency with respect to $M$ and $Y$), and positivity throughout. We consider the following additional assumptions below. 

\stepcounter{assumption}

\begin{intassumption}[Parallel counterfactual trends absent treatment]\label{asmpt:pt0}
    \begin{align*}
        \E[Y_{i1}(0)-Y_{i0}(0) \mid M = m, X = x] &= \E[Y_{i1}(0) - Y_{i0}(0) \mid M = 0, X = x], \qquad m =0,\dots,k \\
    \end{align*}
\end{intassumption}

\begin{intassumption}[Parallel counterfactual trends under treatment]\label{asmpt:pt1}
    \begin{align*}
        \E[Y_{i1}(m)-Y_{i0}(0) \mid M = m, X = x] &= \E[Y_{i1}(m) - Y_{i1}(0) \mid M = 0, X = x], \qquad m = 0,\dots,k
    \end{align*}
\end{intassumption}

\noindent The standard identifying assumption used in difference-in-differences is given by assumption \ref{asmpt:pt0} when $M$ is binary. We simply extend this assumption to account for different treatment levels, as in \cite{callaway2021difference}. We also state a stronger assumption \ref{asmpt:pt1}: that the trends for each treatment group equal the trends that would have occurred for the non-treated group were it treated, allowing us to further identify the average treatment effect on the controls.

\begin{intassumption}[A-M ignorability]\label{asmpt:amign1}
    \begin{align*}
        A &\perp M(1) \mid X
    \end{align*}
\end{intassumption}

\begin{intassumption}[Y-M ignorability]\label{asmpt:amign2}
    \begin{align*}
        Y(m) &\perp M(1) \mid X, A = 0, \qquad m = 0,\dots,k  
    \end{align*}
\end{intassumption}

\noindent A-M ignorability states that $A$ and $M(1)$ are independent conditional on $X$, and Y-M ignorability states that $M(1)$ is independent of $Y(m)$ conditional on $X$ and $A = 0$.

\begin{proposition}\label{prop:did}
    Consider the observed data functional $\psi^{\mathrm{DiD}}(x)$. If consistency, positivity, and assumption \ref{asmpt:pt0} holds, then

    \begin{align*}
         \psi^{\mathrm{DiD}}(x) &= \sum_{m>0} \E[Y_{i1}(m) - Y_{i1}(0) \mid M = m, A = 1, x]\Pr(M = m \mid A = 1, x). 
    \end{align*}

    \noindent This equals a weighted average of conditional average treatment effects among the treated group, with weights proportional to the conditional probability of each treatment value among the treated.
    
    If, in addition, assumptions \ref{asmpt:0} and \ref{asmpt:2a} hold,

    \begin{align*}
         \psi^{\mathrm{DiD}}(x) &= \E[Y_{i1}(M(1)) - Y_{i1}(0) \mid A = 1, x].         
    \end{align*}
    
    \noindent This equals a coarsened conditional average treatment effect on the treated.

    Alternatively, if assumptions \ref{asmpt:0}, \ref{asmpt:2a}, \ref{asmpt:pt0}, \ref{asmpt:pt1}, \ref{asmpt:amign1}, and \ref{asmpt:amign2} hold, then

    \begin{align*}
        \psi^{\mathrm{DiD}}(x) &= \E[Y_{i1}(M(1)) - Y_{i1}(0) \mid x].
    \end{align*}

    This equals a coarsened conditional average treatment effect.    
\end{proposition}

Proposition \ref{prop:did} tells us that under parallel-trends, positivity, and consistency, $\psi^{\mathrm{DiD}}(x)$ returns a weighted average of conditional treatment effects on the treated (CATT), which can also be interpreted as a stochastic intervention where $M$ is randomly drawn from $\P(M \mid A = 1, x)$ among treated units \cite{vanderweele2013causal}. This requires no conception of the $(A, M)$ sequential randomization framework, and is similar to the ignorability result from \cite{hettinger2025causal}, where we have replaced the ignorability assumption with parallel-trends, though note that the latter gives us a weighted average of CATEs, rather than CATTs. If we are also willing to accept the two-stage randomization framework, this same quantity is equal to a coarsened CATT. Finally, to equate this to a coarsened CATE, we may invoke two stronger assumptions: both the stronger version of parallel-trends given by assumption \ref{asmpt:pt1}, and the ignorability assumptions given by \ref{asmpt:amign1} and \ref{asmpt:amign2}.

\begin{proof}
    First, note that under parallel-trends absent treatment, 

    \begin{align*}
        \E[Y_{i1}(0) - Y_{i0}(0) \mid A = 0, x] &= \E[Y_{i1}(0)-Y_{i0}(0) \mid M = m, x],\qquad  m = 1,\dots, k,
    \end{align*}

    \noindent and averaging both sides over $\Pr(M = m \mid A = 1, x)$ we obtain,

    \begin{align*}
        \E[Y_{i1}(0)-Y_{i0}(0) \mid A = 1, x] &= \E[Y_{i1}(0)-Y_{i0}(0) \mid A = 0, x].
    \end{align*}

    \noindent Then,

    \begin{align*}
        \psi^{\mathrm{DiD}}(x) &= \E[Y_{i1}(M)-Y_{i0}(0) \mid A = 1, x] - \E[Y_{i1}(0)-Y_{i0}(0) \mid A = 1, x] \\
        &= \E[Y_{i1}(M)-Y_{i1}(0) \mid A = 1,x] \\
        &= \sum_{m > 0}\E[Y_{i1}(m)-Y_{i1}(0) \mid M = m, A = 1, x]\Pr(M = m \mid A = 1, x) \\
        &= \sum_{m > 0}\E[Y_{i1}(m)-Y_{i1}(0) \mid M(1) = m, A = 1, x]\Pr(M(1) = m\mid A = 1, x) \\
        &= \E[Y_{i1}(M(1))-Y_{i1}(0) \mid A = 1, x] \\
        &= \E[Y_{i1}(M)-Y_{i1}(0) \mid A = 1, x] 
    \end{align*}

    \noindent where the first equality used the previous result, parallel-trends absent treatment, and consistency; the second by simplifying terms; the third by iterating expectations over $M$; the fourth by consistency; the fifth line by summing over the values of $m$; and the final line by consistency. This gives the first two results in Proposition \ref{prop:did}.
    
    Next, we consider instead parallel-trends under treatment:

    \begin{align*}
        \psi^{\mathrm{DiD}}(x) &= \sum_{m>0}\E[Y_{i1}(m)-Y_{i0}(0) \mid M = m, A = 1, x]\Pr(M = m \mid A = 1, x) - \E[Y_{i1}(0)-Y_{i0}(0) \mid M = 0, x] \\ 
        &= \sum_{m>0}\E[Y_{i1}(m)-Y_{i0}(0) \mid A = 0,x]\Pr(M = m \mid A = 1, x) - \E[Y_{i1}(0)-Y_{i0}(0) \mid M = 0, x] \\
        &= \sum_{m>0}\E[Y_{i1}(m)-Y_{i1}(0) \mid A = 0, x]\Pr(M(1) = m \mid A = 0, x)  \\
        &= \sum_{m>0}\E[Y_{i1}(m)-Y_{i1}(0) \mid A = 0, M(1) = m, x]\Pr(M(1) = m \mid A = 0, x) \\
        &= \E[Y_{i1}(M(1))-Y_{i1}(0) \mid A = 0, x] 
    \end{align*}
    
    \noindent where the first equality follows by the fact that $m > 0 \implies A = 1$ and consistency, the second by parallel-trends under treatment, the third by consistency and sequential ignorability and simplifying terms, the fourth by sequential ignorability, and the fifth by averaging over the values of $M(1) = m$. Taking the union of all invoked assumptions implies that $\psi^{\mathrm{DiD}}(x) = \E[Y_{i1}(M(1)) - Y_{i1}(0) \mid x]$.
\end{proof}

\subsubsection{Difference-in-differences: ITE estimator}

    In Section \ref{ssec:extensionscoarsening} we propose applying unit-level difference-in-differences estimators, $\hat\psi_i^{\mathrm{DiD}}$, as starting estimates for the ITE in the case of treatment coarsening. However, we observed that $\hat\psi_i^{\mathrm{DiD}}$ is a better starting estimate of $\psi_i = Y(M(1))-Y(0)$ for treated units relative to untreated units. To illustrate this, consider again the unit-level sample analogue to the conditional parallel-trends assumption given previously by equation \ref{eqn:2a}, where we further assume that $\hat Y_{i1}(m) = Y_{i1}(m)$. In other words,

    \begin{assumption}[Unit-level conditional parallel-trends]\label{asmpt:ulpt}
    \begin{align}
        Y_{i1}(m) &= Y_{i0}+\left[\sum_{j=1}^N\I(X_j = X_i, M_j = m)\right]^\inv\left[\sum_{j=1}^N \I(X_j = X_i, M_j = m)(Y_{j1}-Y_{j0})\right], \qquad\forall i, m = 0, \dots, k,
    \end{align}        
    \end{assumption}

    \noindent noting that we can always take $X = \emptyset$. 

    \begin{proposition}\label{prop:lumpite}
        Under assumption \ref{asmpt:ulpt} and consistency, for a treated unit $i$,

        \begin{align*}
            \hat \psi_i^{\mathrm{DiD}} = \psi_i = Y_{i1}(M_i) - Y_{i1}(0).
        \end{align*}

        \noindent By contrast, for an untreated unit $i$, under the same assumptions as well as the exclusion restriction, we have that,

        \begin{align*}
            \hat \psi_i^{\mathrm{DiD}} = \sum_{m > 0} (Y_{i1}(m)-Y_{i1}(0))\Pr_N(M = m \mid A = 1, X = x).
        \end{align*}

        \noindent This estimand corresponds to a unit-level stochastic intervention that draws $m$ from the conditional distribution of treated states \cite{vanderweele2013causal}.        
    \end{proposition}

The result in Proposition \ref{prop:lumpite} is consistent with the Proposition \ref{prop:did} given in Appendix \ref{appsec:didcoarsening}, though applied to a single unit estimator in a finite sample framework.

    \begin{proof}
     For a treated unit $i$ with $X_i = x$, assumption \ref{asmpt:ulpt} implies that,
    
    \begin{align*}
        Y_{i1}(0) &= Y_{i0}+\left[\sum_{j=1}^N(1-A_j)\I(X_i = x)\right]^\inv\left[\sum_{j=1}^N (1-A_j)\I(X_i = x)(Y_{j1}-Y_{j0})\right]. 
    \end{align*}
    
    \noindent Therefore, $\hat\psi^{\mathrm{DiD}}_{i} = \psi_{i, M(1)} = Y_i(M) - Y_i(0)$, where $Y_i(M) = Y_i$. On the other hand, for an untreated unit $i$ with $X_i = x$, we have that

    \begin{align*}
        \hat Y_{i1}(M_i(1)) &=Y_{i0}(0) +\left[\sum_{j=1}^NA_j\I(X_j = x)\right]^\inv\left[\sum_{j=1}^N A_j\I(X_j = x)(Y_{j1}-Y_{j0})\right].
    \end{align*}

    \noindent Now notice that,
    
    \begin{align*}
    \sum_{j=1}^N A_j\I(X_j = x)(Y_{j1}-Y_{j0}) = \sum_{m>0}\sum_{j=1}^N\I(X_j = x, M_j = m)(Y_{j1}(m)-Y_{j0}(0)),        
    \end{align*}
    
    \noindent so, by definition of $\Pr_N(M = m \mid A = 1, X = x)$, we have

    \begin{align*}
        \hat Y_{i1}(M(1)) &= Y_{i0}(0) \\
        &+ \left[\sum_{m > 0}\sum_{j=1}^N\I(X_j = x, M_j = m)\right]^{-1}\left[\sum_{j=1}^N\I(X_j = x, M_j =m)(Y_{j1}(m)-Y_{j0}(0))\right]\Pr_N(M = m \mid A = 1, X = x) \\
        &= Y_{i0}(0) + \sum_{m > 0}(Y_{i1}(m)-Y_{i0}(0))\Pr_N(M = m \mid A = 1, X = x) \\
        &= \sum_{m > 0}Y_{i1}(m)\Pr_N(M = m \mid A = 1, X = x).
    \end{align*}

    \noindent Applying consistency and subtracting $Y_{i1}$ gives the result.    
    \end{proof}

    Proposition \ref{prop:lumpite} illustrates that for untreated units, under assumption \ref{asmpt:ulpt}, $\hat\psi_i^{\mathrm{DiD}}$ corresponds to a stochastic intervention, where we draw $M$ randomly from the conditional distribution of treated states, and where $Y_{i1}(0) = Y_{i1}$. Unless $Y_{i1}(M_i(1))$ is somehow exactly equal to this weighted average of $Y_{i1}(m)$, this estimate will never return $Y_{i1}(M_i(1))$. Of course, we never believe that these assumptions will hold exactly, which is why we apply the proposed sensitivity framework regardless. Nevertheless, Proposition \ref{prop:lumpite} illustrates why this estimator may return a heuristically ``better'' starting estimate for treated versus untreated units.
\newpage

\subsection{ITE bounds: the unit-level difference-in-differences estimator}\label{appsec:bounds}

We now study the properties of the procedure proposed in Section \ref{sec:bounding} to bound the ITE under a binary intervention $A$, assuming we have data from time periods $-m,\cdots,1$, where treatment occurs when $t = 1$. We first posit the following model for the trends in the potential outcomes and the treatment effects, 

\begin{align}
\label{eqn:pt0}Y_{it}(0) - Y_{it-1}(0) &= \delta_t + \epsilon_{it,0},\\
\label{eqn:pt1}Y_{it}(1) - Y_{it}(0) &= \I(T = t)\psi_i= \I(t = T)(\psi + \epsilon_{it,1}),
\end{align}    

\noindent where we make the following assumptions throughout:

\begin{enumerate}
    \item $Y_{it} = \I(T = 1)A_i Y_{it}(1) + (1-A_i) Y_{it}(0)$
    \item $N^{-1}\sum_{i=1}^N\epsilon_{it,a}\stackrel{a.s.}\to 0, \qquad a = 0, 1$\label{cond:3}
    \item $N^{-1}\sum_{i=1}^N A_i \stackrel{a.s.}\to p, \quad 1 > p > 0$
    \item $N \ge \sum_{i=1}^N A_i > 0$
\end{enumerate}

\noindent The first assumption is equivalent to the consistency assumptions discussed previously, and also precludes anticipatory effects so that $Y_{it}=Y_{it}(0)$ for all units when $t < 1$. The second assumption states that the empirical mean of the error terms converge almost surely to zero, a kind of normalization that ensures that the average of these terms is centered at zero in the limit. This holds, for example, by the SLLN when these error terms are i.i.d., but also under weaker assumptions.\footnote{We remain agnostic as to exactly which conditions hold that guarantee this convergence.} Finally, we assume that the empirical mean of the treatment indicators converges to some constant between zero and one, effectively meaning that there is some positive probability of treatment assignment, and that there is at least one treated and control unit for any given $N$ we consider (implicitly meaning that $N \ge 2$), in order to ensure that the quantities we consider below are well defined.

\begin{remark}\label{rmk:0}
In a repeated sampling framework, assuming that $\E[\epsilon_{it,0} \mid A_i] = 0$ gives the standard parallel-trends assumption that identifies the average treatment effect on the treated. Meanwhile, $\E[\epsilon_{it,1} \mid A_i] = 0$ implies a parallel-trends assumption that further identifies both the average treatment effect on the controls and the average treatment effect. 
\end{remark}

While we will assume distributional properties of the error terms $\epsilon_{it,a}$, to define our causal effect, $\psi_i$, we emphasize that we are conditioning on these errors. That is, nothing is random in our setting. However,  we may at times invoke a repeated sampling framework either to elucidate concepts (as in Remark \ref{rmk:0}) or to derive worst-case bounds that hold with probability close to one.

Finally, we define some additional notation. For any variable $O_i$, we let $\bar O_a = \left[\sum_{i=1}^N \I(A_i = a)\right]^\inv \sum_{i=1}^N \I(A_i = a) O_i$. We also define the following term,

\begin{align*}
\tilde\epsilon_{it,a} = \epsilon_{it,a}-\left[\sum_{j=1}^N\I[(1-A_j) = A_i]\right]^\inv\left[\sum_{j=1}^N\I[(1-A_j) = A_i]\epsilon_{jt,a}\right].
\end{align*}

Last, we define the set 

\[
\tilde\tau_i:= \{\tau_i^{1-A_i}: \psi_i \in [\hat\psi_i^{\mathrm{DiD}} - \tau_i^{1-A_i}, \hat\psi_i^{\mathrm{DiD}} + \tau_i^{1-A_i}]\}.
\]

\noindent In other words, $\tilde\tau_i$ is the set of values of $\tau_i^{1-A_i}$ where the true ITE $\psi_i$ is captured by the bounds given by $\hat\psi_i^{\mathrm{DiD}}$ and a ``valid value'' or choice of $\tau_i^{1-A_i}$.

We now consider how different assumptions about the error terms $\tilde\epsilon_{it,a}$ may inform valid choices of $\tau_i^{1-A_i}$ that allow us to bound $\psi_i$ using $\hat\psi_i^{\mathrm{DiD}}$ , defined in \eqref{eqn:did}, with (near) certainty. We first consider the case where these errors are bounded, and then consider more general error terms, though with alternative distributional assumptions. In the second case, we also discuss at length how pre-treatment data can inform valid choices of $\tau_i^{1-A_i}$.

\subsubsection{Bounded errors}\label{appssec:twotime}

Proposition \ref{prop:1a} considers the case where $\epsilon_{it,a}$ are bounded to derive valid values of $\tau_i^{1-A_i}$. 

\begin{proposition}\label{prop:1a}
    Assume that $\lvert\epsilon_{it,a}\rvert \le \zeta$ . For $i \in \mathcal T$, under equation \eqref{eqn:pt0}, $q_0 \in \tilde\tau_i$ for $q_0 \ge 2\zeta$. For $i \in \mathcal C$, under equations \eqref{eqn:pt0} and \eqref{eqn:pt1}, $q_1 \in \tilde\tau_i$ for $q_1 \ge 4\zeta$.
\end{proposition}

\begin{corollary}\label{corr:1a}
    If in addition to \eqref{eqn:pt0}, we have 
    
    \begin{align}    
    \lim_{N\to\infty}\left[\sum_{i=1}^NA_i\right]^{-1}\left[\sum_{i=1}^NA_i\epsilon_{it,0}\right] \stackrel{a.s.}\to 0\label{eqn:lim0}        
    \end{align}

    \noindent then $q_0 \in \tilde\tau_i$ for $q_0 \ge \zeta + o(1)$ for $i \in \mathcal T$. Similarly, if in addition to \eqref{eqn:pt0} we have both \eqref{eqn:lim0} and
    
    \begin{align}
\lim_{N\to\infty}\left[\sum_{i=1}^NA_i\right]^{-1}\left[\sum_{i=1}^NA_i\epsilon_{it,1}\right] \stackrel{a.s.}\to 0\label{eqn:lim1}   
    \end{align}    
    we have $q_1 \in \tilde\tau_i$ when $q_1 \ge 2\zeta + o(1)$ for $i \in \mathcal C$.
\end{corollary}

\begin{remark}
    Equation \ref{eqn:lim0} from Corollary \ref{corr:1a} would be implied by parallel-trends and i.i.d. error terms $\epsilon_{it,0}$, for example, though would be violated if there were covariates associated with both treatment assignment and the outcome trends absent treatment. By contrast, equation \eqref{eqn:lim1} is arguably much stronger, as it effectively precludes any covariates that drive effect heterogeneity from being correlated with treatment assignment.
\end{remark}

\begin{remark}\label{rmk:1}
    We have assumed that the errors are bounded; however, when the errors are bounded because the outcome itself is bounded, we will have both that $\lvert\epsilon_{it,a}\rvert \le \zeta$ and $\lvert\epsilon_{it,1}+\epsilon_{it,0}\rvert\le \zeta$. In this case, the worst-case bounds for the error on $\hat\psi_i^{\mathrm{DiD}}$ for the treated and untreated units are equal and given by the stated results for the treated units.
\end{remark}

Intuitively, Proposition \ref{prop:1a} shows how we can use $\hat\psi_i^{\mathrm{DiD}}$ and a bounded error assumption to derive values of $\tau_i^{1-A_i}$ that are guaranteed to bound $\psi_i$. Moreover, Proposition \ref{prop:1a} shows that the worst-case bounds on $\psi_i$ are twice as large for untreated units as they are for the treated units, unless, as in Remark \ref{rmk:1}, the error bounds are implied by the outcome bounds. We will continue to see this pattern -- that the bounds on the ITE for untreated units are larger than for treated units -- under alternative assumptions, reflecting the additional uncertainty required to predict $Y_{it}(1)$ relative to $Y_{it}(0)$.

\subsubsection{Other errors}\label{appssec:bounding}

We now consider the case where the error terms are not necessarily bounded. However, in contrast to Proposition \ref{prop:1a}, which only required a boundedness assumption on $\epsilon_{it,a}$, and no notions of randomness, we approach this section instead by making stronger probability statements and distributional assumptions with respect to $\epsilon_{it,a}$, explicitly treating these errors as random and expressing the bounds in probability form. However, we may always use the results below to obtain bounds that hold with probability arbitrarily close to one, so that, conditional on the errors, these bounds are (almost) guaranteed to contain $\psi_i$. We can therefore think of these as derivations that imply worst-case bounds that hold for any particular realizations of $\epsilon_{it,a}$, though for any particular realization of these errors we can theoretically improve upon these bounds. 

We also consider at length how observed pre-treatment data from $t = -m, \dots, 0$ may be used to inform valid values of $\tau_i^{1-A_i}$ in this setting (though the results can also be applied when the errors are bounded). For $t < 1$, we denote the (observable) pre-treatment errors $\tilde\epsilon_{it,0} =Y_{it} - \hat{Y}_{it}$, where $\hat{Y}_{it} = Y_{it-1} + \left[\sum_{j=1}^N\I(1-A_j = A_i)\right]^\inv\left[\sum_{j=1}^N\I(1-A_j = A_i)(Y_{jt}-Y_{jt-1})\right]$, and let $\boldsymbol{\tilde\epsilon}_i = (\tilde\epsilon_{i-m+1,0},\dots,\tilde\epsilon_{i0,0})^\top$. 

We first state Lemma \ref{lemma1}, which allows us to construct confidence intervals at level $1 - \eta$ using a known bound on the MSE of an ITE estimator. We then discuss how this MSE bound may be informed using $\boldsymbol{\tilde\epsilon}_i$, thereby implying valid values of $\tau_i^{1-A_i}$.  

\begin{lemma}\label{lemma1}
    Define $\hat\psi_i = \psi_i + \epsilon_i$, where $\psi_i$ is some fixed value, $\E[\epsilon_i] = \mu_i$ and $\V(\epsilon_i) = \sigma_i^2 >0$. Assume that we know some constant $C$ satisfying $\mu_i^2 + \sigma_i^2 \le C$. Additionally, assume that for any $\eta > 0$, we know some value of $C_\eta$ such that,

    \begin{align}
        \Pr\left(\frac{\lvert \epsilon_i-\mu_i\rvert}{\sigma_i} > C_{\eta}\right)\le \eta.\label{eqn:prob}
    \end{align}

    \noindent Letting $C_\eta'=\sqrt{C_\eta^2 + 1}$, we have that,

    \begin{align*}
        \Pr\left(\frac{\lvert \epsilon_i\rvert}{\sqrt{2C}}> C_\eta'\right) \le \Pr\left(\frac{\lvert \epsilon_i-\mu_i\rvert}{\sigma_i} > C_{\eta}\right)\le\eta.
    \end{align*}
\end{lemma}

\begin{remark}
We may directly apply Lemma \ref{lemma1} to bound either $\tilde\epsilon_{i1,0}$ for treated units, or $\tilde\epsilon_{i1,0}+\tilde\epsilon_{i1,1}$ for untreated units -- that is, to determine a valid value of $\tau_i^{1-A_i}$, in the case where $C$ is some known bound on the MSE of $\hat\psi_i^{\mathrm{DiD}}$. Specifically, we would choose some $\eta$ arbitrarily close to zero, and let $C_0$ be the value of $C_\eta$ that satisfies equation \eqref{eqn:prob}.  We may then set $\tau_i^{1-A_i} = \sqrt{2C(1+C_0)^2}$, where we know that $\tau_i^{1-A_i} \in \tilde\tau_i$ with probability greater than $1-\eta$.
\end{remark}

\begin{remark}\label{rmk:2}
In practice we are unlikely to know some value $C$ that bounds the MSE of $\hat\psi_i^{\mathrm{DiD}}$, so we cannot directly apply Lemma \ref{lemma1}. However, for untreated units, we may be able to inform $C$ using $\boldsymbol{\tilde\epsilon}_i$. For example, $\|\boldsymbol{\tilde\epsilon}_i\|_2$ may be viewed as an estimate of $\sqrt{\mu_i^2 + \sigma_i^2}$, under the implicit assumption that the distribution of the errors $\tilde\epsilon_{it,0}$ is stationary over time. 

In practice, however, we may have few pre-treatment time periods to estimate this quantity well, justifying the use of other norms, such as the $L_\infty$ norm, to claim $\omega\|\boldsymbol{\tilde\epsilon}_i\|^2 > \mu_i^2+\sigma_i^2$ for some $\omega > 0$.  We may then directly apply Lemma \ref{lemma1}, letting $C = \omega\|\boldsymbol{\tilde\epsilon}_i\|^2$ . Applied to our proposed sensitivity framework, this sets $Z_0 = \sqrt{2\omega(1+C_0)^2}$, satisfying $Z_0\|{\tilde\epsilon}_i\| \ge \lvert\epsilon_{i1,0}\rvert$, which holds with probability arbitrarily close to one for any realization of $\epsilon_{i1,0}$. 
\end{remark}

Remark \ref{rmk:2} shows how we can apply Lemma \ref{lemma1} in combination with pre-treatment data to determine valid values of $\tau_i^0$, and therefore bounds on the ITE, for treated units. We next consider bounding the ITE for untreated units using $\boldsymbol{\tilde\epsilon}_i$. To do this, we relate the distributions of $\tilde\epsilon_{i1,1}$ and $\tilde\epsilon_{i1,0}$ in terms of how the means and the variance scale, thereby giving an expression for $Z_1$ in terms of a known $Z_0$ in the case where these errors are multivariate normal, given by Proposition \ref{prop:shift}. 

\begin{proposition}\label{prop:shift}
    Assume 

    \begin{align*}
    (\tilde\epsilon_{i1,0},\tilde\epsilon_{i1,1})\mid A_i =0 \sim \mathcal{N}\left((\mu_{i1,0},\mu_{i1,1}),\begin{pmatrix}
        \sigma_{i1,0}^2 & \sigma_{i1,10} \\ \sigma_{i1,10} & \sigma_{i1,1}^2 \end{pmatrix}\right),        
    \end{align*}    
    
    \noindent and consider the case where $\mu_{i1,1} = \alpha\mu_{i1,0}$ and $\sigma_{i1,1}^2 = \gamma\sigma_{i1,0}^2$ for some $\alpha \in \mathbb R, \gamma\ge 0$. Assume that we know the following for some $\omega > 0$:

    \begin{align}
        \sqrt{\mu_{i1,0}^2 + \sigma_{i1,0}^2} \le \omega\|\boldsymbol{\tilde\epsilon}_i\|
    \end{align}

    \noindent for some norm $\|.\|$. Then,

    \begin{align}
        \sqrt{(\mu_{i1,0} + \mu_{i1,1})^2 + (\sigma_{i1,0}^2 + \sigma_{i1,1}^2 + 2\sigma_{i1,10})} \le C_{\alpha\gamma}\omega\|\boldsymbol{\tilde\epsilon}_i\|,
    \end{align}

    \noindent where $C_{\alpha\gamma} = \sqrt{\max\{(1+\alpha)^2, (1 + \sqrt{\gamma})^2\}}$.
\end{proposition}

\begin{remark}
Proposition \ref{prop:shift}, together with Lemma \ref{lemma1}, implies that we can choose $Z_1 = C_{\alpha\gamma}Z_0$, with $Z_0$ defined as before, to obtain bounds on the ITE for untreated units that hold with probability arbitrarily close to one under the assumed mean and variance scaling. Importantly, we use the assumed normality of these error terms (and therefore their sum) to know that the $C_\eta$ that holds for equation \eqref{eqn:prob} is valid for both $\tilde\epsilon_{i1,0}$ and $\tilde\epsilon_{i1,0} + \tilde\epsilon_{i1,1}$ under $C = Z_0\|\boldsymbol{\tilde{\epsilon}_i}\|$, and $C = C_{\alpha\gamma}Z_0\|\boldsymbol{\tilde{\epsilon}_i}\|$, respectively from Lemma \ref{lemma1}. These results therefore justify using $\hat\psi_i^{\mathrm{DiD}}$ combined with the pre-period errors $\boldsymbol{\tilde\epsilon}_i$ to bound the ITE for untreated as well as treated units. 
\end{remark}

\begin{remark}
We briefly remark on two special cases. First, consider the case when the distributions of $\tilde\epsilon_{i1,0}$ and $\tilde\epsilon_{i1,1}$ are equal (i.e. $\alpha = 1$ and $\gamma = 1$). Proposition \ref{prop:shift} implies that in this case $C_{\alpha\gamma} = 2$, yielding $Z_1 = 2Z_0$, similar to the result we obtain when we only assume that these errors are bounded by the same constant $\zeta$ in Proposition \ref{prop:1a}. 

Second, consider the case where there is no effect heterogeneity, $\alpha = \gamma = 0$. This yields $C_{\alpha\gamma} = 1$, and implies that $Z_1 = Z_0$. This makes intuitive sense, since if $\epsilon_{i1,1} = 0$, then $\lvert\epsilon_{it,1}+\epsilon_{it,0}\rvert = \lvert\epsilon_{it,0}\rvert$ -- that is, there is no additional uncertainty introduced relative to the bound on $\tilde\epsilon_{i1,0}$. 
\end{remark}

Lemma \ref{lemma1} and Proposition \ref{prop:shift} again illustrate the greater uncertainty required to choose valid values of $\tau_i^1$ for untreated units relative to $\tau_i^0$ for treated units, as also seen in Proposition \ref{prop:1a}. In short: we will generally want to choose $Z_1 > Z_0$ when deriving bounds on the ITE for untreated versus treated units.

\subsubsection{Proofs}

\begin{proof}[Proof of Proposition \ref{prop:1a} and Corollary \ref{corr:1a}]
First, consider the case where $A_i = 1$. Then $\hat\psi^{\mathrm{DiD}}_i = Y_{i1} - \hat Y_{i1}(0) = \psi_i +Y_{i1}(0)-\hat{Y}_{i1}(0)$, where $Y_{i1}(1)= Y_{i1}$ by consistency. Moreover,

\begin{align*}
\hat Y_{i1}(0) &= Y_{i0}(0) + \left[\sum_{i =1}^N(1 - A_i)\right]^{-1} \left[\sum_{i =1}^N (1-A_i)(Y_{i1}(0) - Y_{i0}(0)) \right] \\
&= Y_{i1}(0) - \delta_t -\epsilon_{i1,0} + \left[\sum_{i =1}^N(1 -A_i)\right]^{-1}\sum_{i =1}^N(1-A_i)(\delta_t + \epsilon_{i1,0}) \\
&= Y_{i1}(0) - \underbrace{\left(\epsilon_{i1,0} - \left[\sum_{i = 1}^N(1 - A_i)\right]^{-1}\sum_{i = 1}^N(1-A_i)\epsilon_{i1,0}\right)}_{\tilde\epsilon_{i1,0}}
\end{align*}

\noindent so that $Y_{i1}(0) - \hat Y_{i1}(0) = \tilde\epsilon_{i1,0}$, and therefore $\hat\psi^{\mathrm{DiD}}_i = \psi_i + \tilde\epsilon_{i1,0}.$ In this model, $\tau_i^{\star, 0} = \lvert\tilde\epsilon_{i1,0}\rvert \le 2\zeta$ by assumption; moreover, under \eqref{eqn:lim0}, $\lvert\tilde\epsilon_{i1,0}\rvert \le \zeta + o(1)$, noting that condition \ref{cond:3} together with \eqref{eqn:lim0} implies convergence of this final term among the untreated units.

Now consider the case where $A_i = 0$. Then $\hat\psi_i^{\mathrm{DiD}} = \hat Y_{i1}(1)-Y_{i1} = \psi_i + \hat Y_{i1}(1)-Y_{i1}(1)$, where $Y_{i1}(0) = Y_{i1}$ by consistency. Then by \eqref{eqn:pt0} and \eqref{eqn:pt1} we have that,

\begin{align*}
    Y_{i1}(1)-Y_{i0}(0) &= \delta_1 + \psi + \epsilon_{i1,1} + \epsilon_{i1,0}.
\end{align*}

\noindent and therefore,

\begin{align*}
    \hat Y_{i1}(1) &= Y_{i0}(0) + \left[\sum_{i=1}^N A_i\right]^{-1} \left[\sum_{i=1}^N (Y_{i1}(1) - Y_{i0}(0))A_i \right] \\
    &= Y_{i1}(0) - \delta_1 - \epsilon_{i1,0} + \left[\sum_{i=1}^N A_i\right]^{-1}\left(\sum_{i=1}^N(\psi + \delta_1 + \epsilon_{i1,1} + \epsilon_{i1,0})A_i\right) \\
    &= Y_{i1}(1) - \underbrace{\left(\epsilon_{i1,1} + \epsilon_{i1,0} - \left[\sum_{i=1}^N A_i\right]^{-1}\sum_{i=1}^N (\epsilon_{i1,1} + \epsilon_{i1,0})A_i\right)}_{\tilde\epsilon_{i1,1}+\tilde\epsilon_{i1,0}}
\end{align*}

\noindent so that $\hat\psi_i^{\mathrm{DiD}} = \psi_i - (\tilde\epsilon_{i1,1}+\tilde\epsilon_{i1,0})$. We thus conclude that $\tau_i^{\star, 1} = \lvert\tilde\epsilon_{i1,1}+\tilde\epsilon_{i1,0}\rvert \le 4\zeta$, and if both equations \eqref{eqn:lim0} and \eqref{eqn:lim1} hold, we have that $\lvert\tilde\epsilon_{i1,1}+\tilde\epsilon_{i1,0}\rvert \le 2\zeta + o(1)$.
\end{proof}

\begin{proof}[Proof of Lemma \ref{lemma1}]
    First, notice that

    \begin{align*}
        \Pr\left(\frac{\lvert \epsilon_i\rvert}{\sqrt{2C}}> C_\eta'\right)\le\Pr\left(\frac{\lvert \epsilon_i\rvert}{\sqrt{2(\mu_i^2 + \sigma_i^2)}}> C_\eta'\right),
    \end{align*}
    
    \noindent since, by assumption, $C_\eta'\sqrt{2C} \ge C_\eta'\sqrt{2(\mu_i^2 + \sigma_i^2)}$. 
    
    \noindent Next, we have that $\epsilon_i = (\epsilon_i - \mu_i) + \mu_i$ so that, since $(a + b)^2 \le 2(a^2 + b^2)$,

    \begin{align*}
        \epsilon_i^2 \le 2 \{(\epsilon_i - \mu_i)^2 + \mu_i^2\}.
    \end{align*}

    \noindent Therefore,

    \begin{align*}
       \frac{\epsilon_i^2}{2(\mu_i^2 + \sigma_i^2)} \le \frac{(\epsilon_i-\mu_i)^2 + \mu_i^2}{(\mu_i + \sigma_i)^2} \le \frac{(\epsilon_i - \mu_i)^2}{\sigma_i^2} + 1,
    \end{align*}

    \noindent which follows because $\mu_i^2+\sigma_i^2 \ge \sigma_i^2$ and $\mu_i^2/(\mu_i^2+\sigma_i^2)\le 1$. Therefore,

    \begin{align*}
        \frac{\epsilon_i^2}{2(\mu_i^2 + \sigma_i^2)} -1 \le \frac{(\epsilon_i - \mu_i)^2}{\sigma_i^2},
    \end{align*}

    \noindent which then implies,

    \begin{align*}
        \left\{\lvert\epsilon_i\rvert:\frac{\lvert\epsilon_i\rvert}{\sqrt{2(\mu_i^2 + \sigma_i^2)}} > \sqrt{C_\eta^2 + 1}\right\}\subseteq\left\{\lvert\epsilon_i\rvert:\frac{\lvert\epsilon_i-\mu_i\rvert}{\sigma_i} >C_\eta\right\}.
    \end{align*}

    \noindent Therefore, for $C_\eta' = \sqrt{1 + C_\eta^2}$, putting everything together we have,
    
    \begin{align*}
        \Pr\left(\frac{\lvert \epsilon_i\rvert}{\sqrt{2C}}> C_\eta'\right)\le\Pr\left(\frac{\lvert \epsilon_i\rvert}{\sqrt{2(\mu_i^2 + \sigma_i^2)}}> C_\eta'\right) \le \Pr\left(\frac{\lvert \epsilon_i -\mu_i\rvert}{\sigma_i}> C_\eta\right) \le\eta.
    \end{align*}

    \noindent Note that this result follows from purely deterministic inequalities and makes no distributional assumptions.
\end{proof}

\begin{proof}[Proof of Proposition \ref{prop:shift}]
    The proof follows from noting that $\sigma_{i1,10}\le\sigma_{i1,1}\sigma_{i1,0}$, and that

    \begin{align*}
        (\mu_{i1,0} + \mu_{i1,1})^2 + (\sigma_{i1,0}^2 + \sigma_{i1,1}^2 + \sigma_{i1,10})&\le\mu_{i1,0}^2(1+\alpha)^2 + \sigma_{i1,0}^2(1 + \sqrt{\gamma})^2 \\
        &\le C_{\alpha\gamma}^2(\mu_{i1,0}^2 + \sigma_{i1,0}^2)
    \end{align*}

    \noindent so that when

    \begin{align*}
        \sqrt{\sigma_{i1,0}^2+\mu_{i1,0}^2} \le \omega\|\boldsymbol{\tilde\epsilon}_i\|
    \end{align*}

    \noindent we know that

    \begin{align*}
        \sqrt{\E[(\tilde\epsilon_{i1,0} + \tilde\epsilon_{i1,1})^2]} \le \sqrt{\sigma_{i1,0}^2(1+\sqrt{\gamma})^2+\mu_{i1,0}^2(1+\alpha)^2} \le C_{\alpha\gamma}\omega\|\boldsymbol{\tilde\epsilon}_i\|.
    \end{align*}
\end{proof}

\newpage

\section{Relating the CATE, CDE, and ITE under a linear model}\label{appsec:model}

In this section we illustrate how the CATE, CDE, and ITE differ in state-policy settings using simple linear models in a single cross-section of data, allowing for the kinds of heterogeneity discussed in Section \ref{sec:2}. We consider the following model of state-level potential outcomes with respect to a multi-valued intervention $M$, some effect-modifying covariate of interest $X$, and a second effect modifier $U$, where, for simplicity, we have that $(X_i, U_i)\stackrel{iid}\sim\mathcal N\left(0, \left(\begin{matrix} \sigma^2_x & \sigma_{xu} \\ \sigma_{xu} & \sigma^2_u \end{matrix}\right)\right)$. The potential outcomes $Y_i(m, x)$ are then given by

\begin{align*}
Y_{i}(0,x) &= \delta_0 + \alpha_0 U_i + \beta_0 x + \epsilon_{i,0}, \\
Y_{i}(1,x) &= Y_i(0, x) + \delta_1 + \alpha_1U_i + \beta_1x + \epsilon_{i,1}, \\
Y_{i}(2,x) &= Y_i(0, x) + \delta_2 + \alpha_2U_i + \beta_2x + \epsilon_{i,2},
\end{align*}

\noindent and where $\E[\epsilon_{i,m}\mid U_i, X_i, M_i] = 0$ for all $m$. The conditional independence of the error term could be ensured, for example, in a randomized setting where, for example, $X \perp Y(m, x) \mid U$ for all $m, x$ -- that is, the effect modifier is randomly assigned conditional on $U_i$ -- and where $M \perp Y(m) \mid U, X$ -- that is, where the treatment is randomly assigned conditional on $X$ and $U$.\footnote{Combining these assumptions with SUTVA, we may relate the model of the potential outcomes to the observed outcomes,

\begin{align*}
Y_i = \delta_0 + \alpha_0U_i + \beta_0X_i +\sum_{m=1}^2\I(M_i = m)[\delta_m + \alpha_mU_i + \beta_m X_i + \epsilon_{i,m}] + \epsilon_{i,0},
\end{align*}

\noindent which may be estimated using observed data.} 

To relate this model to our discussion about heterogeneity, assume that $M$ reflects two kinds of NALs, $X$ is some other related state-policy, and $U$ is some average of demographics within a state. When $\alpha_1, \alpha_2$ are non-zero, heterogeneity with respect to $U$ may be due to either (1) individual-level heterogeneity, or (2) contextual factors, noting again that we cannot distinguish between these factors at the aggregate level. When $\beta_1, \beta_2$ are non-zero, heterogeneity is due to variation in the state-policy environment. The error terms $\epsilon_{i,m}$ allow for additional unmodeled sources of heterogeneity independent of $(U_i, X_i)$ which may come from any of the aforementioned sources (individual factors, contextual factors, and the policy environment).

This model quantifies the potential outcomes with respect to both $M$ and $X$ jointly, and we can use it to characterize CATEs, CDEs, and ITEs. Specifically, this model implies the following equalities:

\begin{align*}
\psi_m(x) &= \delta_m + \beta_mx, \\
\tilde\psi_m(x) &= \delta_m + \left(\beta_m + \alpha_m\frac{\sigma_{xu}}{\sigma^2_x}\right)x, \\
\psi_{i,m} &= \delta_m + \alpha_m U_i + \beta_m X_i + \epsilon_{i,m}.
\end{align*}

From this, we can see that the CDE, CATE, and ITE are equal under the following conditions: $\epsilon_{i, m} = 0$, $\alpha_m = 0$, and $x = X_i$. In other words, when there is no unmodeled heterogeneity, no effect modification by $U_i$, and we choose to evaluate the CATE and ITE at $X_i$, the observed covariate value for state $i$. Moreover, this result illustrates that the CATE is generally biased for the CDE, with the bias increasing in the correlation between $X_i$ and $U_i$ and the degree of heterogeneity that exists with respect to $U_i$. In other words, unless $X_i$ and $U_i$ are independent ($\sigma_{xu} = 0$),  $U_i$ does not affect the outcome $(\alpha_m = 0)$, or we consider the point $x = 0$ (or more generally, at $\E[X]$), the CATE does not equal the CDE and has no causal interpretation with respect to the effect modifiers of interest. Finally, neither the CDE nor the CATE equals the ITE, which depends on the specific values of $X_i, U_i$ and the idiosyncratic error term $\epsilon_{i,m}$. This illustrates the so-called ``fundamental problem of causal inference:'' we may observe $Y_i(1)$ or $Y_i(0)$ for any unit $i$, but never both \citep{holland1986statistics}.\footnote{We briefly also note that under this specific framework, we may define a linear projection that equals the true CATE function $\tilde\psi_m(x)$, since we have defined the CATE to be a line with respect to $X$. However, this equality will not generally hold.}

We now consider an analysis that uses the coarsened policy indicator $A$ rather than $M$. Applying the results from Section \ref{ssec:coarsening} and Proposition \ref{prop:rep}, we obtain: 

\begin{align*}
\psi_A(x) &= \sum_{m=1,2}\E[Y(m,x)-Y(0,x)\mid M(1) = m]\Pr(M(1) = m \mid x), \\
\tilde\psi_A(x) &= \sum_{m=1,2}\E[Y(m)-Y(0) \mid M(1) = m, X = x] \Pr(M(1) = m \mid x),\\
\psi_{i, M(1)} &= Y_i(M_i(1))-Y_i(0)
\end{align*}

Again, these expressions will differ as they are functions of underlying quantities that also differ. Overall, this model provides a simple mathematical illustration of how and why the CATE, CDE, and ITE, and their coarsened analogues, will generally differ in state-policy settings.

\newpage

\section{Simulation specifications}\label{appsec:sim}

The simulation in Section \ref{ssec:simulation} uses the models defined in Appendix \ref{appsec:model} with the following parameterizations, though extending it to create a longitudinal dataset as described below. To generate covariates, we draw each variable from a normal distribution where $\mu_x = \mu_u = 0$, and $\sigma^2_x = \sigma^2_u = 1$, and $\sigma_{xu} = 1/8$. We then set the following probabilities of treatment assignment:

\begin{align*}
    \Pr(A = 1 \mid x) &= \Phi(a + x)\\
    \Pr(M(1) = 2 \mid x) &=\Phi(x)
\end{align*}

\noindent for $a = \sqrt{2}\Phi^\inv(2/3)$. This has two important implications: first, that $M(1) \perp A \mid X$, meaning that $\Pr(M(1) = m \mid x) = \Pr(M = m \mid x, A = 1)$. And second, that $M(1) \perp (Y(m,x), Y(m)) \mid X$, implying that $\E[Y(m,x)\mid M(1)=m] = \E[Y(m,x)]$, and $\E[Y(m) \mid M(1) = m, x] = \E[Y(m) \mid x]$. In brief, this means that, for our simulation, $\psi_A(x)$ and $\tilde\psi_A(x)$, the coarsened analogues of the CATE and CDE, are simply weighted averages of $\psi_m(x)$ and $\tilde\psi_m(x)$, respectively, with weights given by the conditional probability of receiving intervention type among the treated group.

Consistent with our discussion of $Y(A=1)=Y(M(1))$, we make Bernoulli draws of $A$ and $M(1)$ separately (noting that $M(1) \in \{1,2\}$), and then set $M = M(1)A$. This results in $\Pr(M = 0) \approx 0.33$, $\Pr(M = 1) \approx 0.26$, $\Pr(M = 2) \approx 0.41$. 

To generate the potential outcomes, we set $\alpha_0 = \beta_0 = 0.5$,  $\alpha_1=\beta_1 = \delta_1 = 1$, and $\alpha_2 = \beta_2 = \delta_2 = -1.5$. However, in contrast to the model specified in Appendix \ref{appsec:model}, which focused on a single cross-section of data, we generate $T = 10$ time-periods, where units are only treated at $t = 10$, so that $Y_{it}=Y_{it}(0)$ when $t < 10$, and at $t = 10$ we apply consistency so that $Y_{iT} = Y_{iT}(M)$. For the model of $Y_{it}(0)$, we draw a common time-varying intercept $\delta_t \sim \mathcal N(0,1)$ that takes the place of $\delta_0$ in the parameterization above. Finally, we also allow the idiosyncratic error terms to vary independently over time, setting $\epsilon_{it,0}\stackrel{iid}\sim\text{Exp}(1)-1$, $\epsilon_{it,1}\stackrel{iid}\sim\text{Exp}(1.5)-2/3$, and $\epsilon_{it,2}\stackrel{iid}\sim\text{Exp}(1.5)-2/3$. 

To generate the estimands in Figure \ref{fig:illustration}, we plot data points from a single dataset of $N = 1,000$, where all of the above parameterizations are the same as described above, except setting $\sigma_{xu} = 1/4$ to better illustrate possible bias between $\psi_m(x)$ and $\tilde\psi_m(x)$.

\newpage

\section{Additional results}\label{app:results}

\subsection{Simulation: full results}

Table \ref{tab:full} displays the full simulation results of each of the five estimators considered: the simple difference-in-differences estimator using the mean absolute pre-period errors $\boldsymbol{\tilde\epsilon}_i$ to inform $\tau_i$ setting $Z_1=Z_0 = 2$, denoted $\tau$;  the infeasible estimator that uses $\tau_i^{1-A_i,\star}$ for the sensitivity parameter, denoted $\tau^\star$; the covariate adjusted DiD estimator using the pre-period deviations to inform $\tau_i$ as before, denoted $\tau(x)$; the infeasible covariate adjusted estimator that uses $\tau_i^{1-A_i,\star}$, denoted $\tau_i^\star(x)$; and finally, the CATE estimator using OLS, described in Section \ref{sec:bounding}. 

The two metrics we consider are displayed under the Statistic column: coverage, or the proportion of states whose bounds / confidence intervals contained the ITE; and power and sign, or the proportion of bounds / confidence intervals that both excluded zero and had correctly signed point estimates. We display these metrics for both treated and control units, separately, as denoted in the respective columns, as well as for different sample sizes ($N$).

\begin{table}[ht]
\centering
\caption{Simulation: full results}
\label{tab:full}
\begin{tabular}{llrrrrrr}
\toprule
\multicolumn{2}{c}{ } & \multicolumn{3}{c}{Treatment} & \multicolumn{3}{c}{Control} \\
\cmidrule(l{3pt}r{3pt}){3-5} \cmidrule(l{3pt}r{3pt}){6-8}
Estimator & Statistic & N = 50 & N = 25 & N = 15 & N = 50 & N = 25 & N = 15\\
\hline
ITE bounds: $\tau$ & Coverage & 0.833 & 0.835 & 0.842 & 0.570 & 0.587 & 0.585\\
ITE bounds: $\tau(x)$ & Coverage & 0.841 & 0.842 & 0.842 & 0.543 & 0.566 & 0.572\\
ITE bounds: $\tau^\star$ & Coverage & 1.000 & 1.000 & 1.000 & 1.000 & 1.000 & 1.000\\
ITE bounds: $\tau(x)^\star$ & Coverage & 1.000 & 1.000 & 1.000 & 1.000 & 1.000 & 1.000\\
CATE: OLS (LP target) & Coverage & 0.544 & 0.653 & 0.685 & 0.516 & 0.637 & 0.694\\
\hline
ITE bounds: $\tau$ & Power and sign & 0.511 & 0.489 & 0.477 & 0.213 & 0.208 & 0.197\\
ITE bounds: $\tau(x)$ & Power and sign & 0.462 & 0.396 & 0.338 & 0.132 & 0.136 & 0.145\\
ITE bounds: $\tau^\star$ & Power and sign & 0.779 & 0.767 & 0.768 & 0.438 & 0.440 & 0.432\\
ITE bounds: $\tau(x)^\star$ & Power and sign & 0.756 & 0.733 & 0.703 & 0.326 & 0.342 & 0.355\\
CATE: OLS (LP target) & Power and sign & 0.324 & 0.252 & 0.226 & 0.110 & 0.080 & 0.087\\
\bottomrule
\end{tabular}
\end{table}

The covariate adjusted estimators generally perform worse than their unconditional analogues, possibly a function of the fact that the model for the covariate adjustment is misspecified due to treatment coarsening. The CATE estimators generally perform worst of all, noting, however, that the target of inference of this estimator is a linear projection of the CATE and not the ITE.

\subsection{Application: average effects}

We estimate average and conditional effects using a difference-in-differences framework, under assumptions \ref{asmpt:pt0} and \ref{asmpt:pt1}. When $X = \emptyset$, the estimand is an ATE that averages over states that did and did not have and/or implement PDMPs in this same time period, as well as states with different levels of providers located within rural communities. We also estimate CATEs that stratify within these variables, letting $D$ be an indicator of having a PDMP in 2013 and 2014, and $R$ be an indicator of states in the top 25\% of providers located within rural communities.\footnote{Additional conditional (mean) independence assumptions with respect to $Y(a,d)$ and $D$ would also allow us to equate these CATEs with corresponding CDEs.} We assume the parallel-trends assumptions hold on the corresponding conditioning sets to achieve causal identification. We estimate each treatment effect using a two-way fixed effects model that includes both unit and time fixed effects, subsetting the data to the relevant conditioning set. 

None of these estimates are statistically significant and the point estimates are all close to zero. 
Table \ref{tab:avgfx} displays the results. Cells containing ``$-$'' denote that the corresponding quantity was not estimable due to data limitations. 

\begin{table}
\centering
\caption{Average and conditional effect estimates}\label{tab:avgfx}
\begin{tabular}{lrrrrrr}
\toprule
Estimand & N & Treated N & Estimates & Std Error & Lower 95\% CI & Upper 95\% CI\\
\midrule
$\E[Y(1)-Y(0)]$ & 50 & 26 & 0.013 & 0.013 & -0.012 & 0.039\\
$\E[Y(1)-Y(0) \mid R = 0]$ & 37 & 21 & 0.007 & 0.013 & -0.019 & 0.033\\
$\E[Y(1)-Y(0) \mid R = 1]$ & 13 & 5 & 0.034 & 0.037 & -0.039 & 0.107\\
$\E[Y(10)-Y(00) \mid D = 0]$ & 45 & 24 & 0.013 & 0.013 & -0.012 & 0.038\\
$\E[Y(10)-Y(00) \mid D = 0, R = 0]$ & 32 & 19 & 0.007 & 0.012 & -0.016 & 0.031\\
$\E[Y(10)-Y(00) \mid D = 0, R = 1]$ & 13 & 5 & 0.034 & 0.037 & -0.039 & 0.107\\
$\E[Y(11)-Y(01) \mid D = 1]$ & 2 & 1 & 0.019 & - & - & -\\
$\E[Y(11)-Y(01) \mid D = 1, R = 0]$ & 2 & 1 & 0.019 & - & - & - \\
$\E[Y(11)-Y(01) \mid D = 1, R = 1]$ & 0 & 0 & - & - & - & -\\
\bottomrule
\end{tabular}
\end{table}
\raggedright

\end{document}